\documentclass[11pt]{article}
\usepackage[margin=1in,nohead, footskip=36pt]{geometry}
\usepackage{tabularx,subfigure}
\usepackage{booktabs}
\usepackage{multirow}
\usepackage{authblk}
\usepackage{mathtools}
\usepackage{algorithm,algorithmic}
\usepackage{xspace}
\usepackage{times}
\usepackage[footnotesize]{caption}

\usepackage[applemac]{inputenc} 
\usepackage{amsmath,amssymb,amsfonts,amsbsy,amsthm}
\usepackage[USenglish]{babel}
\usepackage{enumitem}
\usepackage{tikz}
\usepackage{pgfplots}
\usetikzlibrary{arrows,calc,shapes,backgrounds,patterns}
\tikzstyle{every node}=[draw,circle,fill=white,minimum size=3pt,font=\tiny, inner sep=0]
\tikzstyle{every edge}=[->,shorten >=2pt,shorten <=2pt,>=stealth]
\tikzstyle{terminal}=[draw,circle,fill=black,minimum size=3pt,font=\tiny, inner sep=0]
\tikzstyle{lbl}=[draw=none,fill=none,font=\tiny, outer sep = -3pt]
\tikzstyle{dot}=[densely dotted]
\tikzstyle{aux}=[black]

\usepackage{etoolbox}
\preto\align{\par\nobreak\small\noindent}
\expandafter\preto\csname align*\endcsname{\par\nobreak\small\noindent}

\newcommand{\vargadget}[4]{
\footnotesize
\path (#2,#3) node (v#10) [terminal,label=left:$s_{j_{x_{#4}}}$] {}
to  ++(60:1) node (v#1r1) {}
--  ++(0:1) node (v#1r2) {}
--  ++(0:1) node (v#1r3) {}
    ++(0:1) node (v#1r4) {}
--  ++(0:1) node (v#1r5) {}
--  ++(-60:1) node (v#16) [terminal,label=right:$t_{j_{x_{#4}}}$] {}
--  ++(-120:1) node (v#1l5) {}
--  ++(180:1) node (v#1l4) {}
    ++(180:1) node (v#1l3) {}
--  ++(180:1) node (v#1l2) {}
--  ++(180:1) node (v#1l1) {}
--  ++(120:1) node (v#10) [terminal] {};
\draw[dashed,->] (v#10) -- (v#1r1);
\draw[->] (v#1r1) -- (v#1r2);
\draw[->] (v#1r2) -- (v#1r3);
\draw[dot,->] (v#1r3) -- (v#1r4);
\draw[->] (v#1r4) -- (v#1r5);
\draw[dashed,->] (v#1r5) -- (v#16);
\draw[dashed,->] (v#10) -- (v#1l1);
\draw[->] (v#1l1) -- (v#1l2);
\draw[->] (v#1l2) -- (v#1l3);
\draw[dot,->] (v#1l3) -- (v#1l4);
\draw[->] (v#1l4) -- (v#1l5);
\draw[dashed,->] (v#1l5) -- (v#16);
}
\newcommand{\clausegadget}[4]{
\path (#2,#3) node (sk#1) [terminal,label=left:$s_{j_{#4}}$] {}
 -- ++(0:3) node (tk#1) [terminal,label=right:$t_{j_{#4}}$] {};
\draw[->] (sk#1) -- (tk#1);
}
\newcommand{\clausegadgetun}[4]{
\path (#2,#3) node (sk#1) [terminal,label=left:$s_{j_{#4}}$] {}
 -- ++(0:1.5) node (ik#1) {}
 -- ++(0:1.5) node (tk#1) [terminal,label=right:$t_{j_{#4}}$] {};
\draw[dashed] (sk#1) -- (ik#1);
\draw (ik#1) -- (tk#1);

}

\usepackage{pgfplots}
\usepackage{epstopdf}

\newcommand{\DDP}{$2$DDP}
\usepackage{xspace}  

\newtheorem{definition}{Definition}[section]
\newtheorem{theorem}[definition]{Theorem}
\newtheorem{lemma}[definition]{Lemma}
\newtheorem{remark}[definition]{Remark}
\newtheorem{corollary}[definition]{Corollary}
\newtheorem{assumption}[definition]{Assumption}
\newtheorem{proposition}[definition]{Proposition}

\usepackage{algorithm}
\usepackage{algorithmic}
\usepackage{empheq}

\numberwithin{equation}{section}

\newcommand{\classNP}{\mathsf{NP}}
\newcommand{\classP}{\mathsf{P}}
\newcommand{\classAPX}{\mathsf{APX}}

\newcommand{\R}{\mathbb{R}} 
\renewcommand{\S}{\mathcal{S}}
\newcommand{\N}{\mathbb{N}}

\newcommand{\argmin}{\operatornamewithlimits{argmin}}

\renewcommand{\vec}[1]{{\boldsymbol {#1}}}

\newcommand{\intd}{\text{d}}
\renewcommand{\d}{\partial}

\newcommand{\flows}{\mathcal{F}}
\newcommand{\W}{\mathcal{W}}

\newcommand{\true}{\mathtt{true}}
\newcommand{\false}{\mathtt{false}}

\newcommand{\bte}{\textsc{BringToEquilibrium}\xspace}
\newcommand{\scale}{\textsc{ScaleUniformly}\xspace}

\usepackage[draft,margin]{fixme}
\newcommand{\fmth}[2][\empty]{\ifthenelse{\equal{#1}{\empty}}{\fxnote{TH:
#2}}{\fixme[#1]{TH: #2}}}
\newcommand{\fmmk}[2][\empty]{\ifthenelse{\equal{#1}{\empty}}{\fxnote{MK:
#2}}{\fixme[#1]{MK: #2}}}
\newcommand{\fmmg}[2][\empty]{\ifthenelse{\equal{#1}{\empty}}{\fxnote{MG:
#2}}{\fixme[#1]{MG: #2}}}

\begin{document}

\begin{titlepage}
\title{Complexity and Approximation of the \\Continuous Network Design Problem}
\author[1]{Martin Gairing}
\affil[1]{Department of Computer Science, University of Liverpool, UK}
\author[2]{Tobias Harks}
\affil[2]{School of Business and Economics, Maastricht University, The Netherlands}
\author[3]{Max Klimm}
\affil[3]{Institut f\"ur Mathematik, Technische Universit\"at Berlin, Germany}
\date{}
\maketitle 
\thispagestyle{empty}
\begin{abstract}
We revisit a classical problem in transportation, known as the \emph{continuous (bilevel) network design problem}, CNDP for
short. We are given a graph for which the latency of each edge depends on the ratio of the edge flow and the capacity
installed. The goal is to find an optimal investment in edge capacities so as to minimize the sum of the routing cost of the
induced Wardrop equilibrium and the investment cost for installing the capacity. While this problem is considered as challenging in the literature, its
complexity 
status was still unknown. We close this gap showing that CNDP is strongly $\classNP$-complete and $\classAPX$-hard, both on directed and undirected networks and even for instances with affine latencies.

As for the approximation of the problem, we first provide a detailed analysis for a heuristic studied by 
Marcotte for the special case of \emph{monomial} latency functions (Mathematical Programming, Vol.~34, 1986). Specifically, we 
derive a closed form expression of its approximation guarantee for \emph{arbitrary} sets $\S$ of allowed latency functions. 
Second, we propose a different approximation algorithm and show that it has the same approximation guarantee. 
As our arguably most interesting result regarding approximation, we show that using the better of the two approximation algorithms 
results in a strictly improved approximation guarantee for which we give a closed form expression. For affine latencies, 
e.g., this algorithm achieves a $49/41\approx 1.195$-approximation which improves on the $5/4$ that has been shown 
before by Marcotte. We finally discuss the case of hard budget constraints on the capacity investment.
\end{abstract}

\bigskip

\begin{center}
{\bf Keywords}: Bilevel optimization, Optimization under equilibrium constraints, Network design, Wardrop~equilibrium,
Computational complexity, Approximation algorithms
\end{center}

\end{titlepage}

\maketitle 

\section{Introduction}
\label{sec:intro}
Starting with the seminal works of Pigou~\cite{Pigou20} and Wardrop~\cite{Wardrop52}, the impact
of selfish behavior in congested transportation networks has been investigated
intensively over the past decades.  In Wardrop's basic model of traffic flows, the
interaction between the selfish network users is modeled as a non-cooperative game. This game
takes place in a directed graph with latency functions on the edges and a set of
origin-destination pairs, called \emph{commodities}. Every commodity
has a \emph{demand} associated with it, which specifies the amount of
flow that needs to be sent from the respective origin to the respective
destination. It is assumed that every demand represents a large population
of players, each controlling an infinitesimal small amount of flow,
thus, having a negligible impact on the latencies of others.
The \emph{latency} that a player experiences when traversing an edge is determined by a
non-decreasing latency function of the edge flow on that edge. In practice, latency functions
are calibrated to reflect edge specific parameters such as street length and capacity. One of
the most prominent and popular functions used in actual traffic models are the ones put forward
by the Bureau of Public Roads (BPR)~\cite{Bur64}. BPR-type latency functions are of the form
$S_e(v_e) = t_e \cdot \big( 1 + b_e \cdot (v_e/z_e)^{4} \big)$, where $v_e$ is the edge flow,
$t_e$ represents the free-flow travel time, $b_e>0$ is an edge-specific bias, and $z_e$
represents the street capacity. In a \emph{Wardrop equilibrium} (also called \emph{Wardop flow}), 
every player chooses a minimum-latency
path from its origin to the destination; under mild assumptions on the latency functions this
corresponds to a \emph{Nash equilibrium} for an associated non-cooperative game~\cite{Beckmann56}.

It is well known that Wardrop equilibria can be inefficient in the sense that they do not
minimize the total travel time in the network~\cite{Dubey86}. Prominent examples of this inefficiency include
the famous Braess Paradox~\cite{Braess68}, where improving the network infrastructure by adding
street capacity may result in a Wardrop equilibrium with strictly higher total travel time. This
at the first sight surprising non-monotonic behavior of selfish flows illustrates that
designing networks for good traffic equilibria is an important and non-trivial issue. 

In this paper, we revisit one of the most classical network design problems, termed the 
\emph{continuous (bilevel) network design problem}, CNDP for short, which has been introduced by
Dafermos~\cite{Dafermos68}, Dantzig et al.~\cite{Dantzig79}, and Abdulaal et al.~\cite{Leblanc79}, and was later studied by
Marcotte~\cite{Marcotte85mp}. In this problem, we are given a graph for which the latency of each edge depends on the ratio
of the edge flow and the capacity installed and the goal is to find an optimal investment in edge capacities so as to
minimize the sum of the routing cost of the induced Wardrop equilibrium and the investment cost.
From a mathematical perspective, CNDP  is a bilevel optimization problem (cf.~\cite{ColsonMS07,Luo96}
for an overview),  where in the upper
level the edge capacities are determined and, given these
capacities, in the lower level the flow will settle  into a Wardrop equilibrium.
Clearly, the lower level reaction depends on the first level decision because
altering the capacity investment on a subset of edges may result in revised route
choices by users. 

CNDP has been intensively studied since the late sixties (cf.~\cite{Dafermos68,Magnanti84})
and several heuristic approaches have been proposed since then; see Yang et al.~\cite{Yang98}
for a comprehensive survey. Most of the proposed heuristics are numerical in nature and 
involve iterative computations of relaxations of the problem (for instance
the iterative optimization and assignment algorithm as described in~\cite{Marcotte1992}
and augmented Lagrangian methods
or linearizations of the
objective in the leader and follower problem).
An exception is the work of Marcotte~\cite{Marcotte85mp} who considered several
algorithms based on solutions of associated convex optimization
problems which can be solved in polynomial time~\cite{Groetschel93}. He derives worst-case bounds for his
heuristics
and, in particular, for affine latency functions he devises an approximation algorithm
with an approximation factor of $5/4$. For general monomial latency functions plus a constant (including
the latency functions used by the Bureau of
Public Roads~\cite{Bur64}) he obtains a polynomial time $2$-approximation.

Variants of CNDP have also been considered in the networking literature, see 
\cite{Korilis95,Korilis99,Libman99,BhaskarLS13}. These works, however, 
consider the case where a budget capacity
must be distributed among a set of edges to improve the resulting equilibrium. Most results,
however, only work for  simplified network
topologies (e.g., parallel links) or special latency functions (e.g., $M/M/1$ latency functions).


\paragraph{Our Results and Used Techniques.}
Despite more than forty years of research, to the best of our knowledge, the
computational complexity status of CNDP is still unknown.  We close this gap as we show that 
CNDP is strongly $\classNP$-complete and $\classAPX$-hard, both on directed and undirected networks and even for instances with affine latencies of
the form $S_e(v_e/z_e) = \alpha_e + \beta_e \cdot (v_e/z_e)$, $\alpha_e,\beta_e \geq 0$. For the proof of the
$\classNP$-hardness, we reduce from \textsc{3-SAT}. The reduction has the property that in case that the underlying instance of \textsc{3-SAT} has a solution the cost of an optimal solution is equal to the minimal cost of a relaxation of the problem, in which the equilibrium conditions are relaxed. The key challenge
of the hardness proof is to obtain a lower bound on the optimal solution when the underlying \textsc{3-SAT} instance has
no solution. Our main idea is to relax the equilibrium conditions only partially which enables us to bound the cost of an optimal solution from below
by solving an associated constrained quadratic optimization problem.
With a more involved construction and a more detailed analysis, we can even prove $\classAPX$-hardness of the problem. Here, we
reduce from a symmetric variant of
\textsc{MAX-3-SAT}, in which all literals occur exactly twice.
While all our
hardness proofs rely on instances with an arbitrary number of commodities and respective sinks, we show that for instances
in which all commodities share a common sink, CNDP can be solved to optimality in polynomial
time.

In light of the hardness of CNDP, we focus on approximation
algorithms. We first consider a polynomial time algorithm proposed by Marcotte~\cite{Marcotte85mp}. This algorithm, which we call \bte, first
computes a relaxation
of CNDP by removing the equilibrium conditions.
Then, it reduces the edge capacities individually such that the flow computed in the relaxation becomes a Wardrop
equilibrium. We give a closed form expression of the performance of this algorithm with respect to the set $\S$ of
allowed latency functions. Specifically, we show that this algorithm is a $(1 + \mu(\S))$-approximation, where $\mu(\S)
= \sup_{S \in \S} \sup_{x \geq 0} \max_{\gamma \in [0,1]} \gamma \cdot \bigl(1- S(\gamma\,x)/S(x)\bigr)$. The value $\mu(\S)$ has been used before by Correa et al.~\cite{CorSS04}
and Roughgarden~\cite{Rough02} in the context of price of anarchy bounds
for selfish routing where they showed that the routing cost of a Wardrop equilibrium is no more than a factor of $1/(1-\mu(\S))$ away of the cost of a system optimum. For
the special case that $\S$ is the set of polynomials with non-negative coefficients and maximal degree~$\Delta$, we
derive exactly the approximation guarantees that Marcotte obtained for monomials. As an outcome of our more general
analysis, we further derive that this algorithm is a $2$-approximation
for general convex latency functions and a $5/4$-approximation for concave latency functions.

We then propose a new algorithm which we call \scale. 
This algorithm first computes an optimal solution of the relaxation (as before) and then \emph{uniformly} scales the capacities with a certain parameter $\lambda(\mathcal{S})$
that depends on the class of allowable latency functions $\mathcal{S}$. Based on well-known techniques using variational inequalities (Correa et al.~\cite{CorSS04}
and Roughgarden~\cite{Rough02}), we prove  
that this algorithm also yields a $(1 + \mu(\S))$-approximation. 
As our main result regarding approximation algorithms, we
show that using the better of the two solutions returned
by \bte and \scale yields strictly better approximation guarantees.
We give a closed form expression for the new approximation
guarantee (as a function of $\mathcal{S}$) that, perhaps interestingly, 
depends not only on the well-known value $\mu(\S)$ but also on the argument maximum $\gamma(\S)$
in the definition of $\mu(\S)$. We demonstrate the applicability of this
general bound by showing that it achieves a $9/5$-approximation
for $\mathcal{S}$ containing arbitrary convex latencies. For affine latencies it achieves
a $49/41\approx 1.195$-approximation improving on the $5/4$ of Marcotte.
An overview of our results compared to those of Marcotte can be found in Table~\ref{tab:results} in the appendix.

In the final section we consider the case of 
arbitrary convex 
constraints on the capacity
variables that includes global as well as individual budget constraints on edges.
We show that solving the relaxed problem with removed equilibrium constraints achieves
a trivial approximation ratio of $1/(1-\mu(\S))$ using the well-known price of anarchy
results.  For affine latencies, however, we show that this is essentially best possible
by giving a corresponding hardness result. All proof missing in this extended abstract can be found in the appendix.
   
\paragraph{Further Application.}
Our results have impact beyond the classical application of designing street capacities of road
networks. In the telecommunications networking literature, Wardrop
equilibria appear in networks with \emph{source-routing}, where it is assumed that end-users
choose least-delay paths knowing the state of all available paths. As outlined
in~\cite{Valiant10}, Wardrop equilibria arise even in networks with distributed delay-based 
routing protocols such as OSPF using delay for setting the routing weights. In telecommunication networks, the latency at
switches and routers depends on the installed capacity and has been modeled by BPR-type functions of the form
$S_e(v_e/z_e)= \rho \; (1+0.15\;(v_e/z_e))^4$,
 where $\rho$ represents the propagation delay and $z_e$
 the installed capacity~\cite{Qiu06}. 
 These functions fit into our framework, and 
 our analysis improves the state-of-the-art to a $1.418$-approximation.
Additionally,
 our $9/5$-approximation applies to Davidson latency functions of the
 form $S_e(\frac{v_e}{z_e}) =\frac{v_e}{z_e}/(1-\frac{v_e}{z_e}) = v_e/(z_e-v_e)$, where $z_e$ represents
the capacity of edge~$e$. These functions behave quite similar to
 the frequently used $M/M/1$-delay functions 
of the form $S_e(v_e) = 1/(z_e-v_e)$, cf.~\cite{Korilis95,Roughgarden02}.  
 
\paragraph{Further Related Work.}
Quoting~\cite{Yang98}, CNDP has been recognized to be ``one of the most difficult and challenging problems in
transport'' and there are numerous works 
approaching this problem. In light of the substantial  literature  on
heuristics for CNDP, we refer the reader to the
survey papers~\cite{ColsonMS07,Friesz85,Magnanti84,Yang98}.

While to the best of our knowledge
prior to this work, the complexity status of CNDP was open,
there have been several papers on the complexity
of the \emph{discrete (bilevel) network design problem}, DNDP for short, see~\cite{Lin11,Roughgarden06}.
Given a network with edge latency functions and traffic demands,
a basic variant of DNDP is to decide which edges should be removed from the network to obtain a Wardrop equilibrium in the
resulting sub-network with minimum total travel time. This variant
is motivated by the classical Braess paradox, where removing
an edge from the network may improve the travel time
of the new Wardrop equilibrium. Roughgarden~\cite{Roughgarden06} showed that DNDP is strongly $\classNP$-hard and that there
is no $(\lfloor n/2\rfloor-\epsilon)$-approximation algorithm (unless $\classP = \classNP$), even for single-commodity
instances.  He further showed that for single-commodity instances the trivial algorithm of not removing
any edge from the graph is essentially best possible and achieves a $\lfloor n/2\rfloor$-approximation.
For affine latency functions, the trivial algorithm gives a $4/3$-approximation (even for general networks)
and this is also shown to be
best possible. These results in comparison to ours highlight interesting
differences. While DNDP is not approximable by any constant for convex latencies, for CNDP we give
a $9/5$-approximation. Moreover, all hardness results for DNDP already hold for single-commodity instances, while
for CDNP we show that this case is solvable in polynomial time. 

Bhaskar et al. \cite{BhaskarLS13} studied a variant of CNDP where initial edge capacities
are  given and additional
budget must be distributed among the edges to improve the resulting equilibrium. Among other results they show that the 
problem is NP-complete in single-commodity networks that consist of parallel links in series. This again stands in contrast to our 
polynomial-time algorithm for CDNP for these instances.

\section{Preliminaries}
Let $G = (V,E)$ be a directed or undirected graph, $V$ its set of vertices and $E \subseteq V
\times V$ its set of edges. We are given a set $K$ of \emph{commodities}, where
each commodity~$k$ is associated with a triple $(s_k,t_k,d_k) \in V \times V \times
\R_{>0}$, where $s_k \in V$ is the \emph{source}, $t_k \in V$ the \emph{sink} and $d_k$ the \emph{demand} of
commodity $k$. A multi-commodity flow on $G$ is a collection of non-negative flow vectors $(\vec v^k)_{k \in
K}$ such that for each $k \in K$ the flow vector $\vec v^k = (v_e^k)_{e \in E}$ satisfies the flow
conservation constraints $\smash{\sum_{u \in V : (u,w) \in E} v^k_{(u,w)}} - \smash{\sum_{u \in V : (w,u) \in
E} v^k_{(w,u)} = 0}$ for all $w \in V \setminus \{s_k,t_k\}$ and $\smash{\sum_{u \in V : (s_k,u) \in E}
v^k_{s_k,u}} = \smash{\sum_{u \in V : (u,t_k) \in E} v^k_{u,t_k}} = \smash{d_k}$. Whenever we write $\vec v$
without a superscript $k$ for the commodity, we implicitly sum over all commodities,
i.e., $v_e = \sum_{k \in K} v_e^k$ and $\vec v = (v_e)_{e \in E}$. We call $v_e$ an
\emph{edge flow}. The set of all feasible edge flows will be denoted by $\flows$.

The latency of each edge~$e$ depends on the installed capacity $z_e \geq 0$ and the edge flow
$v_e$ on $e$, and is given by a latency function $S_e : \R_{\geq 0} \to \R_{\geq 0} \cup
\{\infty\}$ that maps $v_e/z_e$ to a latency value $S_e(v_e/z_e)$, where we use the convention
that $S_e(v_e/z_e) = \infty$ whenever $z_e =0$. Throughout this paper, we assume
that the set of allowable latency functions is restricted to some set $\S$ and we impose the following
assumptions on $\S$.

\begin{assumption}
\label{ass}
The set $\S$ of allowable latency functions only contains continuously
differentiable and semi-convex functions $S$ such that the functions $x \mapsto S(x)$ and
$x \mapsto x^2S'(x)$ are strictly increasing and unbounded.
\end{assumption}

Assumption~\ref{ass} is is slightly more general than requiring that all latency functions are stricly increasing and
convex.
For instance, the function $S(x) := \sqrt{x}$ satisfies Assumption~\ref{ass} although it is concave. 

Given a vector of capacities $\vec{z} = (z_e)_{e \in E}$, the latency of each edge $e$ solely
depends on the edge flow $v_e$. Under these conditions, there exists a Wardrop flow
$\vec{v} = (v_e)_{e \in E}$, i.e., a flow in which each commodity only uses paths of minimal
latency. It is well known (see e.g. \cite{Beckmann56,Dafermos80,Smith79b}) that each Wardrop flow is a solution to
the optimization problem $\min_{\vec v \in \flows} \sum_{e\in E} \int_{0}^{v_e} S_e(t/z_e) \,\intd t$, and satisfies the
\emph{variational inequality}
\begin{align}
\label{eq:variational-inequality}
\sum_{e\in E} S(v_e/z_e) (v_e-v_e') \le 0
\end{align}
for every feasible flow $\vec v'\in \flows$.
For a vector of capacities $\vec z$ we denote by $\W(\vec{z})$ the corresponding set of Wardrop
flows $\vec{v}(\vec{z})$.
%
Beckmann et al.~\cite{Beckmann56} showed that  Wardrop flows and optimum flows are related:
\begin{proposition}[Beckmann et al.~\cite{Beckmann56}]
\label{pro:WardropVsOptimum}
Denote by $S_e^*(x)=(xS_e(x))'=S_e(x)+xS'_e(x)$ the \emph{marginal cost function} of edge $e\in E$. Then $\vec v^*$ is an optimum flow 
with respect to the latency functions $(S_e)_{e\in E}$ if and only if it is Wardrop flow with respect
to $(S_e^*)_{e\in E}$.
\end{proposition}
In the continuous (bilevel) network design problem (CNDP) the goal is to buy capacities $z_e$ at
a price per unit $\ell_e > 0$ so as to minimize the sum of the construction cost
$C^Z(\vec v, \vec z) = \sum_{e \in E} z_e\, \ell_e$ and the routing cost $C^R(\vec v, \vec z) = \sum_{e \in E}
S_e(v_e/z_e)\, v_e$ of a resulting Wardrop equilibrium $\vec v$. Observe that $C^R(\vec v, \vec z)$ 
is well defined as, by \eqref{eq:variational-inequality}, it is the same for all Wardrop equilibria with respect to $\vec z$.
Denote  the combined cost by $C(\vec v, \vec z)=C^R(\vec v, \vec z)+C^Z(\vec v, \vec z)$.

\begin{definition}[Continuous network design problem (CNDP)]
Given a directed graph $G = (V,E)$ and for each edge $e$ a latency function $S_e$ and a construction cost $\ell_e>0$, the
continuous network design problem (CNDP) is to determine a non-negative capacity vector $\vec z = (z_e)_{e \in E}$ that
minimizes 
\begin{align}
\tag{CNDP}
\label{nd}
\begin{split}
\min_{\vec{z} \geq 0} \min_{\vec v \in \W(\vec z)} &\sum_{e\in E} \bigl( S_e(v_e/z_e)\;v_e
+z_e\;\ell_e \bigr).
\end{split}
\end{align}
\end{definition}
Relaxing the condition that $\vec v$ is a Wardrop equilibrium in \eqref{nd}, we obtain the
following relaxation of the continuous network design problem:
\begin{align}
\tag{CNDP'}
\label{rnd}
\begin{split}
\min_{\vec{z} \geq 0} \min_{\vec v \in \flows} &\sum_{e\in E} \bigl(
S_e(v_e/z_e)\;v_e
+z_e\;\ell_e \bigr).
\end{split}
\end{align}

Marcotte~\cite{Marcotte85mp} showed that for convex and unbounded latency functions, the relaxed problem \eqref{rnd} can be
solved efficiently by performing $|K|$ independent shortest path computations on the graph $G$, one for each commodity $k
\in K$. The following proposition slightly generalizes his result to arbitrary, not necessarily convex latency functions
that
satisfy Assumption~\ref{ass}.

\begin{proposition}[Marcotte \cite{Marcotte85mp}]
\label{pro:marcotte}
The relaxation \eqref{rnd} can be solved by performing $|K|$ shortest path computation problems in polynomial time.
\end{proposition}

\begin{remark}
To speak about polynomial algorithms and hardness, we need to specify how the instances of CNDP, in particular the latency
functions, are encoded, cf.~\cite{Ahuja93,Groetschel93,Roughgarden06}. While our hardness results hold even if all functions
are linear and given by their rational coefficients, for our approximation algorithms, we require that we can solve
(symbolically) equations involving a latency function and its derivative,
e.g.,~Equation~\eqref{eq:define_gamma}. Without this assumption, we still obtain the claimed approximation guarantees within
arbitrary precision by polynomial time algorithms.
\end{remark}

\section{Hardness}

As the main result of this section, we show that CNDP is $\classAPX$-hard both on directed and undirected networks and even for affine latency functions. The
proof of this result is technically quite involved, and we first show the weaker result that CNDP on directed networks is
$\classNP$-complete.
Due to space constraints, we here only sketch the proof of the $\classNP$-completeness for directed networks and the case that there are
edges with zero latency. For the full proof and the discussion that the problem remains hard, even if no edges with zero
latency are allowed, we refer to the appendix.                                                
                                                                                               
\begin{theorem}
\label{thm:hardness_mc}
The continuous network design problem on directed networks is $\classNP$-complete in the strong sense, even if all latency
functions are affine.
\end{theorem}

\begin{proof}[Sketch of proof]
We reduce from \textsc{3-SAT}. Let a Boolean formula $\phi$ in conjunctive normal form be given and for
$\kappa, \nu \in \N$, let $K(\phi) = \{1,2,\dots,\kappa\}$ and $V(\phi) = \{x_1,x_2,\dots,x_{\nu}\}$ denote the set of its
clauses and variables, respectively. For each variable $x_i \in V(\phi)$, we introduce a \emph{variable commodity} $j_{x_i}$
with
unit demand, and for each clause $k \in K(\phi)$ we introduce a \emph{clause commodity} $j_k$ with unit demand. For each
literal $l$ and each clause $k$, there is a \emph{literal edge} $e_{l,k}$ with latency function
$S_{e_{l,k}}(v_{e_{l,k}},z_{e_{l,k}}) = v_{e_{l,k}}/z_{e_{l,k}}$ and construction cost $\ell_{e_{l,k}} = 1$. Further, for
some $\epsilon>0$ and for each clause $k$, there is a \emph{clause edge} $e_k$ with $S_{e_k}(v_{e_k}/z_{e_k}) =
4+v_{e_k}/z_{e_k}$ 
and construction cost $\ell_{e_k} = (\epsilon/2)^2$.
Every variable commodity $j_{x_i}$ has two feasible paths, one consist of the literal edges $\{e_{x_i,k} : k \in K(\phi)\}$
corresponding to the positive literal $x_i$, the other one consists of the literal edges $\{e_{\bar{x}_i,k} : k \in
K(\phi)\}$ corresponding to the negative literal $\bar{x}_i$. In that way, each route choice of the variable commodities
corresponds to a fractional assignment of the variables. For each clause $k = l_k \vee l_k' \vee l_k''$, the clause
commodity $j_k$ has two feasible paths as well, one consists of the clause edge $e_k$, the other one contains the three
literal edges $e_{l_k,k}$, $e_{l_k',k}$, and $e_{l_k'',k}$. We add some additional edges with zero latency to this path in
order to obtain a network structure, see Figure~\ref{fig:hardness_mc} in the appendix.

Let us first assume that $\phi$ has a solution $\vec y = (y_{x_i})_{x_i \in V(\phi)}$ and let $\bar{\vec y} =
(\bar{y}_{x_i})_{x_i \in V(\phi)}$ be the negation of $\vec y$. Then, an optimal solution to the so-defined instance of
CNDP is follows. For each variable commodity $j_{x_i}$, we buy capacity $1$ on the path consisting of the edges
$\{e_{\bar{y}_{x_i},k} : k \in K(\phi)\}$ and we route the unit demand of variable commodity $j_{x_i}$ over the edges of
that path. For each clause commodity $j_k$, we route the unit demand over the clause edge $e_k$. Using that $\vec y$ is a
solution of $\phi$, we derive that for each clause, there is a literal $l_k^*$ that occurs in $\vec y$, and thus, $l_k^*$
does not occur in
$\bar{\vec y}$. However, this implies that for each clause $k = l_k \vee l_k' \vee l_k''$, at least one of the three literal
edges $e_{l_k,k}$, $e_{l_k',k}$, and $e_{l_k'',k}$ has capacity $0$ and, thus, infinite latency. Thus, the clause
commodity $j_k$ has only one path with finite length and we conclude that the so-defined flow is a Wardrop equilibrium. This
solution has total cost $2\kappa\nu + (4+\epsilon)\kappa$ which can be shown to be minimal as it coincides with
the total cost of the relaxation of the problem without the equilibrium constraints.

If $\phi$ does not admit a solution, we show that each feasible solution has cost strictly larger than $2\kappa\nu +
(4+\epsilon)\kappa$. Assume by contradiction that there is a solution $(\vec v, \vec z)$ with
cost at most $2\kappa\nu + (4+\epsilon)\kappa$. We claim that in $\vec v$, each clause commodity $e_k$
uses its clause edge, i.e., $v_{e_k} >0$. To see this, note that each unit of flow of the clause commodities that is routed
over the three corresponding literal edges contributes at least $6$ to the total cost of a solution while each unit of flow that is
routed over a clause edge contributes at most $(4+\epsilon)$ to the total cost. This implies, that the total cost is at
least $2\kappa\nu + (4+\epsilon)\kappa + (2-\epsilon)$ if one of the clause commodities does not use its clause edge. 
However, since $\phi$ does not admit a solution, we cannot prevent a clause commodity from using three of the corresponding
literal edges without reducing the capacity on at least one of these edges below $1$. Reducing the capacity on the literal
edges below $1$, however, comes at a cost, since the resulting capacities are then strictly smaller than in the relaxation
of the problem. By solving an associated constrained quadratic program, we show that the total cost of any feasible solution
is at least $2\kappa\nu + (4+\epsilon)\kappa + 1/8$, if $\phi$ does not admit a solution.
\end{proof}

With a more involved construction and a more detailed analysis, we can show that CNDP is in fact $\classAPX$-hard. For
this proof, we use a similar construction as in the proof of Theorem~\ref{thm:hardness_mc}. However, instead from
\textsc{3-SAT}, we reduce from a specific variant of \textsc{MAX-3-SAT}, which is $\classNP$-hard to approximate.

\begin{theorem}
\label{thm:hardness_apx}
The continuous network design problem on directed networks is $\classAPX$-hard, even if all latency function are affine.
\end{theorem}

With a similar construction, we can also show $\classAPX$-hardness for CNDP on undirected networks as
well, see Theorem~\ref{thm:hardness_apx_un} in the appendix. For our hardness results, we use instances with different
sinks. In contrast, CNDP can be solved efficiently for networks
with a single sink. 

\begin{proposition}
\label{pro:single_sink}
In networks with only one sink vertex $t$, the continuous network design problem can be solved in polynomial time.
\end{proposition}

\section{Approximation}
Given the $\classAPX$-hardness of the problem, we study the approximation of CNDP. We first
provide a detailed analysis of the approximation guarantees of two different approximation algorithms. Then, as the
arguably most interesting result of this section, we provide an improved approximation guarantee for taking the better of
the two algorithms. The approximation guarantees proven in this section depend on the set $\S$ of allowable cost functions
and are in fact closely related to the \emph{anarchy value} value $\alpha(\S)$ introduced by Roughgarden~\cite{Rough02} and
Correa et al.~\cite{CorSS04}. Intuitively, the anarchy value of a set of
latency functions $\S$ is the worst case ratio between the routing cost of a Wardrop equilibrium and that of a system
optimum of an instance in which all latency functions are contained in $\S$. Roughgarden~\cite{Rough02} and Correa et
al.~\cite{CorSS04} show that $\alpha(\S) = 1/(1-\mu(\S))$, where

\begin{align}
\label{def:mu}
\mu(\S)&=\sup_{S\in \S}\sup_{x \geq 0}
\max_{\gamma \in [0,1]} \gamma \cdot \Bigl(1- \frac{S(\gamma\,
x)}{S(x)}\Bigr).
\end{align}
For a set $\S$ of latency functions, we denote by $\gamma(\S)$ the argmaximum $\gamma$ in \eqref{def:mu} for which
$\mu(\S)$ is achieved.
The following lemma gives an alternative representation of $\mu(S)$ that will be useful in
the remainder of this section.

\begin{lemma}
\label{lem:equivalent}
For a latency function $S$,
\begin{align*}
\sup_{x \geq 0} \max_{\gamma \in [0,1]}
\Bigl\{\gamma\,\Bigl(1-\frac{S(\gamma\,x)}{S(x)}\Bigr)\Bigr\}
= \sup_{x \geq 0} \Bigl\{\gamma \cdot \frac{S'(x)\,x}{S(x)+S'(x)\,x} : S(x) + S'(x)\,x =
S(x/\gamma)\Bigr\}.
\end{align*}
\end{lemma}

\subsection{Two Approximation Algorithms}

The first algorithm that we call \bte (cf.~Algorithm~1) was already proposed by Marcotte
\cite[Section 4.3]{Marcotte85mp} and
analyzed for monomial latency functions. Our contribution is a more
general analysis of \bte that works for arbitrary sets of latency
functions~$\S$, requiring only Assumption~\ref{ass}. The second algorithm, that we call \scale (cf.~Algorithm 2), is a new
algorithm that we introduce in this paper.

For both approximation algorithms, we first compute an optimum solution $(\vec v^*, \vec z^*)$ 
to a relaxation of CNDP without the equilibrium constraints, i.e., we compute a solution $(\vec v^*, \vec z^*)$ to the problem\linebreak
$\min_{\vec z \geq 0} \min_{v\in \flows} \sum_{e\in E} \bigl(S_e(v_e/z_e)\;v_e
+z_e\;\ell_e\bigr)$, which can be done in polynomial time (Proposition~\ref{pro:marcotte}).
Then, in both algorithms, we reduce the capacity vector $\vec z^*$, and determine a 
Wardrop equilibrium for the new capacity vector. The algorithms differ in the way we adjust the 
capacity vector $\vec z^*$. 
While in \bte, we reduce the edge capacities individually such that the optimum solution to the relaxation 
\eqref{rnd} is a Wardrop equilibrium, in \scale, we scale all capacities uniformly by a factor $\lambda$ (cf. line
\ref{l:p}-\ref{l:lambda}) and compute a Wardrop equilibrium for the scaled capacities.

\begin{figure}[tb]
\noindent\begin{minipage}[t]{0.49\textwidth}
\begin{algorithm}[H]
\begin{algorithmic}[1]
\STATE $(\vec v^*, \vec z^*) \gets $ solution to relaxation \eqref{rnd}.
\FORALL{$e\in E$}
\STATE $\delta_e \gets v_e^*/z^*_e$
\STATE $\gamma_e \gets \text{solution to } S_e(\delta_e)\!+\!S'_e(\delta_e)\delta_e=S_e(\frac{\delta_e}{\gamma_e})$ 
\STATE $z_e \gets \gamma_e z_e^*$
\ENDFOR
\RETURN{($\vec v^*, \vec z$)}
\end{algorithmic}
\caption{\bte}\label{alg:bte}
\end{algorithm}
\end{minipage}
\begin{minipage}{0.009\textwidth}
~
\end{minipage}
\begin{minipage}[t]{0.49\textwidth}
\begin{algorithm}[H]
\begin{algorithmic}[1]
\STATE $(\vec v^*, \vec z^*) \gets $ solution to relaxation \eqref{rnd}.
\STATE $p \gets C^R(\vec v^*, \vec z^*)/C(\vec v^*, \vec z^*)$\label{l:p}
\STATE $\lambda \gets \mu(\S)+\sqrt{\mu(\S)\frac{p}{1-p}}$\label{l:lambda}\phantom{\scalebox{0.86}{$\Bigg|$}}
\STATE Compute Wardrop equilibrium $\vec v$\newline with respect to scaled capacities $\lambda \vec z^*$.
\RETURN{($\vec v, \lambda \vec z^*$)}
\end{algorithmic}
\caption{\scale}\label{alg:us}
\end{algorithm}
\end{minipage}
\end{figure}

We first show that the approximation guarantee of \bte is at most $(1+\mu(\S))$. For the proof of this result, we use the
first order optimality conditions for the vector of capacities $\vec v^*$ obtained as a solution to the relaxed problem
\eqref{rnd} in combination with the variational inequalities technique used in the price of anarchy literature (e.g.
Roughgarden~\cite{Rough02} and Correa et
al.~\cite{CorSS04}).  

\begin{theorem}
\label{thm:bring-to-equilibrium}
The approximation guarantee of \bte is at most $(1+\mu(\S))$.
\end{theorem}
\begin{proof}
Let $(\vec v^*,\vec z^*)$ be the relaxed solution computed in the first step of \bte. By the necessary Karush-Kuhn-Tucker
optimality conditions, $(\vec v^*,\vec z^*)$ satisfies
\begin{align}
\label{eq:opt_capacity}
\ell_e=S'_e(v_e^*/z^*_e)(v_e^*/z_e^*)^2, \text{ for all } e \in E \text{ with }z^*_e>0\,.
\end{align} 
Eliminating $\ell_e$ in the statement of the relaxed problem \eqref{rnd} we obtain the following expression for the total
cost of the relaxation: 
\begin{align}
\label{ref}
C(\vec{v^*},\vec{z^*}) &= \sum_{e\in E} \Big(S_e(v_e^*/z_e^*) +S'_e(v_e^*/z_e^*)
(v_e^*/z_e^*) \Big) v_e^*\,.
\end{align}
For each $e \in E$ let $\delta_e = v_e^*/z_e^*$, if $z_e^* > 0$, and $\delta_e = 0$, otherwise. We define a new vector of
capacities $\vec z$ by
$z_e = \gamma_e \cdot z_e^*,e\in E$, where $\gamma_e\in [0,1]$ is a solution to the equation 
\begin{align}
\label{eq:define_gamma}
 S_e(\delta_e)+S'_e(\delta_e)\,\delta_e=S_e(\delta_e/\gamma_e).
\end{align}
By Proposition \ref{pro:WardropVsOptimum},
the flow $\vec v^*$  is a Wardrop flow with respect to $\vec z$. 
We are interested in bounding $C(\vec v^*, \vec z)$. To this end, we calculate
\begin{align}
C(\vec v^*,\vec z)
&\overset{\phantom{MM}}{=} \sum_{e\in E} \left( S_e(\delta_e/\gamma_e)v_e^* +\ell_e\,z_e \right)
\nonumber \overset{\eqref{eq:define_gamma}}{=} \sum_{e\in E} \Bigl(\big(S_e(\delta_e)
+S'_e(\delta_e)\,\delta_e\big)v^*_e+\gamma_e\,\ell_e\,z_e^*\Bigr)
\nonumber\\
&\overset{\eqref{eq:opt_capacity}}{=} \sum_{e\in E}  \Bigl(\bigl(S_e(\delta_e)
+S'_e(\delta_e)\,\delta_e\bigr)v^*_e + \gamma_e\,S_e'(\delta_e)\,\delta_e\,v_e^*\Bigr)\,.
\label{eq:mid_bte}
\end{align}
By \eqref{def:mu},\eqref{eq:define_gamma}, and Lemma \ref{lem:equivalent}, we have
$
\gamma_e\,S_e'(\delta_e)\,\delta_e 
\le \mu(\S) \left( S_e(\delta_e) + S_e'(\delta_e)\,\delta_e \right)$.
Combining this inequality with \eqref{eq:mid_bte}, gives
\begin{align*}
C(\vec v^*,\vec z)
&\overset{\phantom{\eqref{def:mu}}}{\leq} (1+\mu(\S)) \sum_{e\in E}  \Bigl(\bigl(S_e(\delta_e)
+S'_e(\delta_e)\,\delta_e\Bigr)v^*_e \overset{\eqref{ref}}{=} (1+\mu(\S))\,C(\vec v^*,\vec z^*).\qedhere
\end{align*}
\end{proof}

We proceed by showing that \scale achieves the same approximation guarantee of $1 + \mu(\S)$. Recall that \scale first
computes a relaxed solution $(\vec v^*, z^*)$. Then, this relaxed solution is used to compute an optimal scaling factor
$\lambda \leq 1$ with which all capacities are scaled subsequently. The algorithm then returns the scaled capacity vector
$\lambda \vec z^*$ together with a correspond Wardrop equilibrium $v \in \W(\lambda \vec z^*)$.

An (worse) approximation guarantee of $2$ can be infered directly from a bicriteria result of Roughgarden and
Tardos~\cite{Roughgarden02} who showed that for any instance the routing cost of a Wardrop equilibrium is not worse than a
system optimum that ships twice as much flow. This implies that for $\lambda = 1/2$ we have $C(v, \lambda z^*) \leq 2
C(v^*,z^*)$, as claimed. 

For the proof of the following result, we take a different road that allows us to express the approximation
guarantee of \scale as a function of the parameter $p$ defined as the fraction of the total cost $C(\vec v^*, \vec z^*)$ of
the relaxed solution allotted to the routing costs $C^R(\vec v^*, \vec z^*)$. This is an important ingredient for the
analysis of the best-of-two algorithm.
\begin{theorem}
\label{thm:scale-uniformly}
The approximation guarantee of \scale is at most $(1+\mu(\S))$.
\end{theorem}
\begin{proof}
The algorithm first computes an optimum solution $(\vec v^*, \vec z^*)$ of the relaxed problem \eqref{rnd}. Then $p \in
[0,1]$ is defined as the fraction of $C(\vec v^*, \vec z^*)$ that corresponds to the routing cost
$C^R(\vec v^*, \vec z^*)$, i.e., $C^R(\vec v^*,\vec z^*) = \sum_{e\in E} S_e(v_e^*/z_e^*)\,v_e^*
= p\,C(\vec v^*,\vec z^*)$. Now, we define  $\smash{\lambda=\mu(\S)+\sqrt{\mu(\S)\frac{p}{1-p}}}$ and consider the capacity
vector
$\lambda \vec z^*$, in which the capacities of the optimal solution to the relaxation 
are scaled uniformly by $\lambda$. 
Finally, we compute a Wardrop equilibrium with respect to capacities $\lambda \vec z^*$.
Let $\vec v$ the corresponding equilibrium flow. We now bound the routing and
installation
cost of $(\vec v,\lambda \vec z^*)$ separately. For the installation
cost, we obtain
\begin{align*}
C^Z(\vec v,\lambda \vec z^*)
&= \sum_{e\in E}\lambda\, \ell_e\, z_e=\lambda (1-p)\,C(\vec v^*,\vec z)
\end{align*}
and for the routing cost
\begin{align}
C^R(\vec v,\lambda \vec z^*)
&= \sum_{e\in E} S_e\Bigl(\frac{v_e}{\lambda z_e^*}\Bigr)\,v_e \leq \sum_{e\in E} S_e\Bigl(\frac{v_e}{\lambda
z_e^*}\Bigr)\,v_e^*
= p\,C(\vec v^*,\vec z^*) + \sum_{e\in E} \left( S_e\Bigl(\frac{v_e}{\lambda z_e^*}\Bigr)\,v_e^*-             
S_e\Bigl(\frac{v_e^*}{z_e^*}\Bigr)\,v_e^*\right),\label{eq:routing_cost}
\end{align}
where the first inequality uses the variational inequality \eqref{eq:variational-inequality}.
We proceed to bound $S_e(\frac{v_e}{\lambda z_e^*})\,v_e^*- S_e(\frac{v_e^*}{z_e^*})\,v_e^*$
in terms of the routing cost $S_e(\frac{v_e}{\lambda z_e^*})\,v_e$ for that edge $e$. To this end, note that for each edge
$e \in E$ we have 
\begin{align}
\frac{S_e(\frac{v_e}{\lambda z_e^*})v_e^*-
S_e\bigl(\frac{v_e^*}{z_e^*}\bigr)v_e^*}{S_e(\frac{v_e}{\lambda z_e^*})v_e}
&\leq \sup_{S \in \S} \sup_{x,y,z \geq 0} \!\!\frac{S(\frac{y}{\lambda z})x-S(\frac{x}{z})x}{S(\frac{y}{\lambda z})y}
= \sup_{S \in \S} \sup_{x,y \geq 0} \!\!\frac{S(\frac{y}{\lambda})x-S(x)x}{S(\frac{y}{\lambda})y}
= \sup_{S \in \S} \sup_{x,y \geq 0} \!\!\frac{S(y)\,x-S(x)x}{S(y)\lambda y}.\notag
\intertext{This implies $y \geq x$ and we may substitute $x =\gamma\,y$ with $\gamma \in [0,1]$. We then obtain for each
edge $e \in E$ that}
\frac{S_e(\frac{v_e}{\lambda z_e^*})v_e^*-
S_e\bigl(\frac{v_e^*}{z_e^*}\bigr)v_e^*}{S_e(\frac{v_e}{\lambda z_e^*})v_e}
&\leq \sup_{S \in \S} \sup_{y \geq 0} \!\max_{\gamma \in [0,1]} \frac{\gamma S(y)-\gamma S(\gamma\,y)}{\lambda\,S(y)}
= \sup_{S \in \S} \sup_{y \geq 0} \!\max_{\gamma \in [0,1]} \frac{\gamma}{\lambda}\Bigl(1-\frac{S(\gamma\,
y)}{S(y)}\Bigr) = \frac{\mu(\S)}{\lambda}.\label{eq:sup_bound}
\end{align}
Plugging \eqref{eq:sup_bound} in \eqref{eq:routing_cost}, we obtain for the routing cost
$C^R(\vec v,\lambda \vec z^*) \leq p\,C(\vec v^*,\vec z^*)+\frac{\mu(\S)}{\lambda}\,C^R(\vec v,\lambda \vec
z^*)$ or, equivalently, $C^R(\vec v,\lambda \vec z^*) \leq \frac{p}{1-\mu(\S)/\lambda} C(\vec v^*,\vec
z^*)$.
Thus, we can bound the total cost of the outcome of \scale by
\begin{align*}
C(\vec v,\lambda \vec z^*) &=C^R(\vec v,\lambda \vec z^*)+C^Z(\vec v,\lambda \vec z^*) \leq \frac{p}{1-\mu(\S)/\lambda}
C(\vec v^*,\vec z^*) +\lambda(1-p)\,C(\vec v^*,\vec z^*) \\
& = \lambda \Big(\frac{p}{\lambda-\mu(\S)}+1-p\Big)C(\vec v^*,\vec z^*).
\end{align*}
Since
$
\lambda = \mu(\S) + \sqrt{\mu(\S) \frac{p}{1-p}},
$
we obtain
\begin{align}
\frac{C(\vec v,\lambda \vec z^*)}{C(\vec v^*,\vec z^*)} 
\le p + 2\sqrt{p(1-p)\mu(\S)} + \mu(\S)(1-p)
= \bigl(\sqrt{p} + \sqrt{\mu(\S)(1-p)}\bigr)^2.
\label{eq:mid_scale}
\end{align}
Elementary calculus shows that 
$\bigl(\sqrt{p} + \sqrt{\mu(\S)(1-p)}\bigr)^2$ attains its maximum at $p=\frac{1}{1+\mu(S)}$.
Substituting this value into \eqref{eq:mid_scale} gives $C(\vec v,\lambda \vec z^*) / C(\vec v^*,\vec z^*) \le 1+ \mu(\S)$,
as claimed.
\end{proof}

For particular sets $\S$ of latency functions, we compute upper bounds upper bounds on $\mu(\S)$ in order to obtain an
explicit upper bound on the approximation guarantees of \bte and \scale. We then obtain the following corollary of
Theorem~\ref{thm:bring-to-equilibrium} and Theorem~\ref{thm:scale-uniformly}.

\begin{corollary}
\label{cor:simple}
For a set $\S$ of latency functions satisfying Assumption~\ref{ass}, the approximation guarantee 
of \bte and \scale 
is at most 
\begin{itemize}
\itemsep-0.2em
\item[(a)] $2$, without further requirements on $\S$.
\item[(b)] $\frac{5}{4}$, if $\S$ contains concave latencies only,
\itemsep-0.4em
\item[(c)] $1 +\frac{\Delta}{\Delta+1}\bigl(\frac{1}{\Delta+1}\bigr)^{1/\Delta}$, if $\S$ 
contains only polynomials with non-negative coefficients and degree at most $\Delta$, 
i.e., every $S\in\S$ is of the form $S(x) = \sum_{j=0}^{\Delta}a_j x^j$ with $a_j\geq 0$ for all $j$. 
\end{itemize}
\end{corollary}

%

\subsection{Best-of-Two Approximation}

In this section we show that although both \bte and \scale achieve an approximation guarantee of $(1+\mu(\S))$ taking the
better of the two algorithms we obtain a strictly better performance guarantee.

The key idea of the proof is to extend the analysis of the \bte algorithm in order to express its approximation guarantee
as a function of the parameter $p$ that measures the proportion of the routing cost in the total cost of a relaxed
solution. This allows us to determine the
worst-case $p$ for which the approximation guarantee of the both algorithm is maximized. 

\begin{theorem}
\label{thm:better}
Taking the better solution of \bte and \scale has an approximation guarantee of at most
$
\frac{(\gamma(\S) + \mu(\S) +1)^2}{(\gamma(\S) + \mu(\S) + 1)^2 - 4\mu(\S)\gamma(\S)},
$
which is strictly smaller than $1+\mu(\S)$.
\end{theorem}

\begin{proof}
Recall from \eqref{eq:mid_scale} that the approximation guarantee of the algorithm \scale is 
\linebreak$\smash{\bigl(\sqrt{p} + \sqrt{\mu(\S)(1-p)}\bigr)^2}$, where $\smash{p=C^R(\vec v^*, \vec z^*)/C(\vec v^*, \vec z^*)}$. 
We extend our analysis of \bte using this parameter $p$.
With the notation in Theorem~\ref{thm:bring-to-equilibrium}, by \eqref{eq:mid_bte}, \bte returns a feasible
solution
$(\vec v^*,\vec z)$ with
\begin{align*}
C(\vec v^*,\vec z)
&= \sum_{e\in E}  \Bigl(\bigl(S_e(\delta_e)
+S'_e(\delta_e)\,\delta_e\bigr)v^*_e + \gamma_e\,S_e'(\delta_e)\,\delta_e\,v_e^*\Bigr)
=
p\,C(\vec v^*,\vec z^*)+\sum_{e\in E} S'_e(\delta_e)\,\delta_e\,v^*_e (1+\gamma_e)\\
&\le
p\,C(\vec v^*,\vec z^*)+(1+\gamma(\S)) \sum_{e\in E} S'_e(\delta_e)\,\delta_e\, v_e^*
=
p\,C(\vec v^*,\vec z^*)+(1+\gamma(\S))(1-p)\,C(\vec v^*,\vec z^*)\\
&=
\big( 1+ \gamma(\S)(1-p)\big) \, C(\vec v^*,\vec z^*).
\end{align*}
Thus, by
taking the best of the two heuristics, we obtain an approximation guarantee of
\begin{align*}
\max_{p\in (0,1)}\min\left\{1 + \gamma(\S)(1-p), \Bigl(\sqrt{p} +
\sqrt{\mu(\S)(1-p)}\Bigr)^2\right\}.
\end{align*}
The maximum of this expression is attained for 
\begin{align}
p=p^* := \frac{(\gamma(\S)-\mu(\S)+1)^2}{(\gamma(\S)-\mu(\S)+1)^2 + 4\mu(\S)}
\label{eq:p_star}
\end{align}
which yields the claimed improved upper bound (cf. Lemma~\ref{lem:maxof2} in the appendix for details).
%
%
\end{proof}

It is not necessary to run both approximation algorithms to get this approximation
guarantee.  After computing the optimum solution to the relaxation \eqref{rnd}, we can determine 
the value for $p=C^R(\vec v^*, \vec z^*)/C(\vec v^*, \vec z^*)$ and proceed with \scale if
$p\le p^*$ (cf.  \eqref{eq:p_star}) and with \bte otherwise.

Fort particular sets $\S$ of latency functions, we evaluate $\mu(\S)$ and $\gamma(\S)$ and obtain the following corollary of
Theorem~\ref{thm:better}.

\begin{corollary}
\label{cor:better}
For a set $\S$ of latency functions satisfying Assumption~\ref{ass}, the approximation guarantee in Theorem~\ref{thm:better}
is at most 
\begin{itemize}
\itemsep-0.2em
\item[(a)] $\frac{9}{5}$, without further requirements on $\S$,
\item[(b)] $\frac{49}{41}\approx 1.195$, if $\S$ contains concave latencies only.
\itemsep-0.3em
\item[(c)] $1 +\frac{4\Delta(\Delta+1)}{2(2\Delta+1)(\Delta+1)^{1+1/\Delta} + (\Delta+1)^{2(1+1/\Delta)} + 1}$, if $\S$
contains only polynomials with non-negative coefficients and degree at most $\Delta$, 
i.e., every $S\in\S$ is of the form $S(x) = \sum_{j=0}^{\Delta}a_j x^j$ with $a_j\geq 0$ for all $j$. 
\end{itemize}
\end{corollary}

\section{Conclusion}

We reconsidered the classical continuous network design problem (CNDP). To the best of our knowledge, we established the
first hardness results for CNDP. Specifically, we have shown the $\classAPX$-hardness of CNDP both on directed and
undirected networks and even if all latency functions are affine. We then turned to the approximation of the problem.
First, we provided a thorough analysis of an algorithm proposed and studied by Marcotte~\cite{Marcotte85mp} for monomial
latency functions. 
We showed a general approximation guarantee depending on the set of allowed cost functions which is related to the \emph{anarchy value} of the set of cost functions.
Second, we proposed and studied a different approximation algorithm that turned out to provide
the same approximation guarantee. As our arguably most interesting result concerning approximation, we then showed that
taking the best of the two algorithms, we can guarantee a strictly better approximation factor.  

In the transportation literature, further variants of CNDP have been investigated. One such example 
are situations in which the network designer is only interested in minimizing total travel time but
investments are restricted, e.g., by budget constraints. More generally, suppose there is a convex function
$g:\R^m\rightarrow \R^k$, $k\in \N$ such
that for any feasible solution $\vec{z}$ the condition $g(\vec{z})\leq \vec{0}$ must be satisfied. 
The function $g$, for instance, can represent edge-specific budget constraints  $\ell_ez_e\leq B_e$ for $e\in E$ and/or  a
global budget constraint $\sum_{e\in E}\ell_e z_e\leq B$. We arrive at the following budgeted continuous network design
problem (bCNDP):
\begin{align}
\tag{bCNDP}
\label{budgets}
\min_{\vec z \geq 0}\min_{\vec v \in \W(\vec z)}   \sum_{e\in E} S_e(v_e/z_e)\;v_e
 \;\; s.t.: g(\vec{z})\leq \vec{0}.
\end{align}

Using existing results from the price of anarchy literature (Roughgarden~\cite{Rough02} and Correa et
al.~\cite{CorSS04}), we can show that there is a $4/3$-approximation for affine latencies and assuming $\classP \neq
\classNP$, for any $\epsilon>0$, there is no polynomial time approximation algorithm with a performance
guarantee better than $4/3-\epsilon$, see Theorem~\ref{thm:budget} in the appendix.
For proving the lower bound, we use edge-specific budget constraints and mimic a construction from Roughgarden \cite{Roughgarden06}. It is an
interesting open problem whether such a lower bound can also be achieved if we allow only a global budget constraint.  

\newpage
\bibliographystyle{abbrv}
\bibliography{master-bib}

\clearpage
\appendix
\setcounter{section}{1}
\section*{Appendix}

\subsection*{Missing Material of Section~\ref{sec:intro}}

\begin{table}[h]
\caption{Approximation guarantees of the algorithms \bte, \scale, and the best of the two for convex latency
functions, concave latency functions and sets of polynomials with non-negative coefficients depending on the maximal degree
$\Delta$. The approximation guarantees stated for convex latency functions even hold for sets of semi-convex
latency functions as in Assumption~\ref{ass}. For \bte, the approximation guarantees marked with $(^\star)$ have been
obtained before in \cite{Marcotte85mp}.\\[-2.5\baselineskip]}
\label{tab:results}
\begin{center}
\footnotesize
\begin{tabular*}{0.62\linewidth}{
@{}l
@{\extracolsep{0ex}}   c
@{\extracolsep{\fill}} r
@{\extracolsep{4ex}}   r
}
\toprule
& & \multicolumn{2}{c}{\raisebox{0.1cm}{\rule{1.3cm}{0.5pt}}\,\,Approximation
guarantees\,\,\raisebox{0.1cm}{\rule{1.3cm}{0.5pt}} }\\
\multicolumn{2}{l}{\multirow{2}{*}{Functions}} & \multicolumn{1}{c}{\bte}  & \multirow{2}{*}{Better of the two}\\
& & \multicolumn{1}{c}{\scale}\\
\midrule
concave   & & $5/4 = 1.25\phantom{0^\star}$ & $49/41 \approx 1.195$\\
convex & & $2\phantom{.000^\star}$       & $9/5 = 1.8\phantom{00}$\\
\midrule
polynomials & $\Delta$\\
&0 & $1\phantom{.000^\star}$  & $1\phantom{.000}$\\
&1/4 &$3381/3125 \approx 1.082\phantom{^\star}$   & $\approx 1.064$\\
&1/3 &$283/256 \approx 1.105\phantom{^\star}$  & $\approx 1.083$\\
&1/2 &$31/27 \approx 1.148\phantom{^\star}$  & $1849/1657 \approx  1.116$\\
&1 & $5/4 = 1.25\phantom{0}^\star$ 
  & $49/41 \approx 1.195$\\
&2 & $1 + \frac{2}{9}\sqrt{3} \approx 1.385^\star$ 
  & $\frac{311}{479} + \frac{180}{479} \sqrt{3} \approx 1.300$\\
&3 & $1 + \frac{3}{16} \sqrt[3]{4^2} \approx 1.472^\star$ 
  & $\approx 1.369$\\
&4 & $1 + \frac{4}{25} \sqrt[4]{5^3} \approx 1.535^\star$ 
  & $\approx 1.418$\\
& $\infty$ & $2\phantom{.000}^\star$ 
  & $9/5 = 1.8\phantom{00}$\\
\bottomrule\\[-30pt]
\end{tabular*}
\end{center}
\end{table}

\subsection*{Proof of Proposition~\ref{pro:marcotte}}

\begin{proof}
As the latency of all edges diverges to $\infty$ as the capacity approaches $0$ we obtain $z_e>0$ if and only if $v_e
>0$ for all edges $e \in E$. The Karush-Kuhn-Tucker conditions of the relaxed problem \eqref{rnd} imply that
\begin{align*}
\frac{\d}{\d z_e} \sum_{e\in E} \bigl(S_e(v_e/z_e)\;v_e+z_e\;\ell_e \bigr) = 0,
\end{align*}
or, equivalently, $\ell_e = (v_e/z_e)^2 S_e'(v_e/z_e)$ for all $e \in E$ with
$z_e>0$. Using that $x^2S'_e(x)$ is non-decreasing and unbounded, for each $e \in E$ there is a  solution to the equation
$x^2\,S_e'(x) = l_e$ which we denote by $u_e$. Since $\ell_e>0$, we
derive that $u_e>0$ as well. By definition, $u_e$ is the unique optimal ratio of $v_e/z_e$ for
edge $e$ with $z_e >0$ in an optimal solution of \eqref{rnd}. Substituting $z_e = v_e/u_e$ in
\eqref{rnd}, we obtain the equivalent mathematical problem
\begin{align*}
\min_{\vec v \in \flows} \sum_{e\in E} \bigl( S_e(u_e) + \ell_e/u_e\bigr) v_e,
\end{align*}
which can be solved by performing $|K|$ independent shortest path computations, one for each
commodity $k \in K$.
\end{proof}

\subsection*{Proof of Theorem~\ref{thm:hardness_mc}}

\begin{proof}
CNDP lies in $\classNP$ as a vector of capacities $\vec z$ is a polynomial certificate. Given $\vec z$, we can compute in
polynomial time a corresponding Wardrop equilibrium and the total cost $C(\vec v, \vec z)$.

To show the $\classNP$-hardness of the problem, we reduce from \textsc{3-SAT}. Let
$\phi$ be a Boolean formula in conjunctive normal form. We denote the set of variables and
clauses of $\phi$ with $V(\phi)$ and $K(\phi)$, respectively, and set $\nu = |V(\phi)|$ and $\kappa = |K(\phi)|$. The set
$L(\phi)$ of literals of $\phi$ contains for each variable $x_i \in V(\phi)$ the positive literal $x_i$ and the negative
literal $\bar{x}_i$, i.e., $L(\phi) = \{x_i \in V(\phi)\} \cup \{\bar{x}_i : x_i \in
V(\phi)\}$. In the following, we will associate clauses with the set of literals that they
contain.

We now explain the construction of a continuous network design problem based on $\phi$ that
has the property that, for some $\epsilon \in (0,1/8)$, an optimal solution has total cost less or equal to
$(4+\epsilon)\kappa + 2\kappa\nu$ if and only if $\phi$ has a solution. Let $\epsilon \in (0,1/8)$ be arbitrary.
For each clause $k \in K(\phi)$, we introduce a \emph{clause edge} $e_k$ with latency
function $S_{e_k}(v_{e_k}/z_{e_k}) = 4 + v_{e_k}/z_{e_k}$ and construction cost $\ell_{e_k} = (\epsilon/2)^2$. For each
literal $l \in L(\phi)$ and each clause $k \in K(\phi)$, we introduce a \emph{literal edge} $e_{l,k}$ with latency function
$S_{e_{l,k}}(v_{e_{l,k}}/z_{e_{l,k}}) = v_{e_{l,k}}/z_{e_{l,k}}$ and cost $\ell_{e_{l,k}} = 1$. We denote the set of clause
edges and literal edges by $E_K$ and $E_L$, respectively.

For each variable $x_i \in V(\phi)$, there is a \emph{variable commodity} $j_{x_i}$ with source $s_{j_{x_i}}$, sink
$t_{j_{x_i}}$ and demand $d_{j_{x_i}} = 1$. This commodity has two feasible paths, one path uses exclusively the literal
edges $\{e_{x_i,k} : k \in K(\phi)\}$ that correspond to the non-negated variable $x_i$, the
correspond to the negated variable $\bar{x}_i$. In that way, each feasible path of the variable
commodity $j_{x_i}$ corresponds to a $\true$/$\false$ assignment of the variable $x_i$. For each clause $k = l_{k} \vee
l_{k}' \vee l_{k}''$, we introduce a \emph{clause commodity} $j_k$ with source $s_{j_k}$, sink $t_{j_k}$ and demand
$d_{j_k} = 1$. The clause commodity may either choose its corresponding clause edge $e_k$ or the corresponding literal edges
that occur in $k$, i.e., $e_{{l_k},k}$, $e_{{l'_k},k}$, and $e_{l''_k,k}$. For notational convenience, we set $E_k =
\{e_{l_k,k},e_{l_k',k},e_{l_k'',k}\}$.
We add some additional edges with latency $0$ to obtain a network; see Figure~\ref{fig:hardness_mc} where these edges are
dashed. Note that the problem remains $\classNP$-hard, even if we do not allow edges with zero latency, see
Remark~\ref{rem:zero_latency} after this proof.

\begin{figure}[tb]
\begin{center}
\begin{tikzpicture}[x=0.7cm,y=0.7cm]
\path[use as bounding box] (-1,-1.5) rectangle (21.5, 3.5);

\vargadget{1}{0}{0}{1}
\vargadget{2}{7.5}{0}{2}
\node at (14,0) [draw=none,fill=none] {$\dots$};
\vargadget{n}{15.5}{0}{\nu}
\clausegadget{1}{2}{3}{1}
\clausegadget{2}{8}{3}{2}
\node at (13.5,3) [draw=none,fill=none] {$\dots$};
\clausegadget{n}{16}{3}{\kappa}
\draw[->,aux,dashed] (sk1)  to[bend right=10] (v1r1);
\draw[->,aux,dashed] (v1r2) .. controls (5,3) and (7,-2) .. (v2l1);
\draw[->,aux,dashed] (v2l2) .. controls (13,-4) and (15,2.7) .. (vnr1);
\draw[->,aux,dashed] (vnr2) .. controls (15,3) and (6,1) .. (tk1);
\draw[->,dot,aux] (sk2) to[bend right=10] ++(-130:1) node
(dummy1)[draw=none,fill=none] {};
\draw[<-,dot,aux] (tk2) to[bend left=10] ++(-50:1) node
(dummy2)[draw=none,fill=none] {};
\draw[->,dot,aux] (skn) to[bend right=10] ++(-130:1) node
(dummy1)[draw=none,fill=none] {};
\draw[<-,dot,aux] (tkn) to[bend left=10] ++(-50:1) node
(dummy2)[draw=none,fill=none] {};
\end{tikzpicture}
\end{center}
\caption{\label{fig:hardness_mc}
Network used to show the hardness of the continuous network design
problem. Clause~$1$ is equal to $x_1 \vee \bar{x}_2 \vee x_{\nu}$. Dashed edges have zero latency.}
\end{figure}
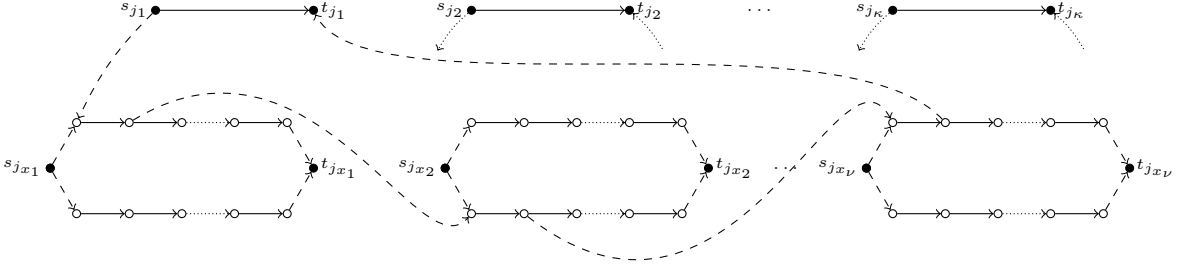

First, we show that an optimal solution of the so-defined instance of the continuous network
design problem $P$ has total cost less or equal to $(4+\epsilon)\kappa + 2\kappa\nu$, if
$\phi$ has a solution. To this end, let $\vec y = (y_{x_i})_{x_i \in V(\phi)}$ be a solution of
$\phi$. Then, a feasible solution of $P$ is as follows: For each \emph{positive} literal
$x_i$ that is selected in the solution $y_i$, we buy capacity $1$ for the corresponding
\emph{negative} literal edges $\{e_{\bar{x}_i,k} : k \in K(\phi)\}$, and vice versa. Formally, we set
\begin{align*}
z_{a_{l,k}} &=
\begin{cases}
1, &\text{if } l = x_i \text{ and } y_{x_i} = {\tt false},\\
1, &\text{if } l = \bar{x}_i \text{ and } y_{x_i} = {\tt true},\\
0, &\text{otherwise.}
\end{cases}
\end{align*}

For each clause edge $e_k$, $k \in K(\phi)$, we buy capacity $2/\epsilon$. This particular capacity vector $\vec z =
(z_e)_{e \in E}$ implies that each variable commodity $j_{x_i}$ has a unique path of finite length, i.e., the path using the
edges corresponding to the \emph{negation} of the corresponding literal in $\vec y$. Using that $\vec y$ is a solution of
$\phi$, we further obtain that for each clause commodity $j_k$ at least one
of the edges in $E_k$ has capacity zero and, thus, infinite latency. This implies that, in the unique
Wardrop equilibrium, the demand of each clause commodity $j_k$ is routed along the corresponding
clause edge $e_k$. For the total cost of this solution, we obtain
\begin{align}
C(\vec v, \vec z) &= \sum_{e \in E_K} \bigl( (4+v_e/z_e)v_e + (\epsilon/2)^2 z_e \bigr) + \sum_{e \in E_L} \bigl(
(v_e/z_e)v_e + z_e\bigr)\notag\\
&= \sum_{e \in E_K} \bigl( (4+\epsilon/2) + (\epsilon/2) \bigr) + \frac{1}{2}\sum_{e \in E_L} \bigl(1 + 1\bigr)=
(4+\epsilon)\kappa
+ 2\kappa\nu.\label{eq:opt_value}
\end{align}
Hence, an optimal solution has cost not larger than \eqref{eq:opt_value} if $\phi$ has a
solution.

We proceed to prove that the total cost of an optimal solution are strictly larger than \eqref{eq:opt_value} if $\phi$
does \emph{not} admit a solution. Let $\vec z = (z_e)_{e \in E}$ be an optimal solution of
$P$ and let $\vec v = (v_e)_{e \in E}$ be a corresponding Wardrop flow. We distinguish two cases.

\emph{First case:} $v_{e_k} >0$ for all $k \in K(\phi)$, i.e., each clause commodity $j_k$ sends flow over the corresponding
clause edge $e_k$.

Before we prove the thesis for this case, we need some additional notation. For the
Wardrop flow $v_e$ on edge $e \in E$, let $v_e^V$ and $v_e^K$ denote the flow on $e$ that is due
to the variable commodities and the clause commodities, respectively. 

We claim that there is a clause $\tilde{k} \in K(\phi)$, $\tilde{k} = l_{\tilde{k}} \vee l_{\tilde{k}}' \vee
l_{\tilde{k}}''$ such that the flow of the variable commodities on each of the corresponding literal edges in $E_{\tilde{k}}
= \{e_{l_{\tilde{k}},\tilde{k}},e_{l_{\tilde{k}}',\tilde{k}},e_{l_{\tilde{k}}'',\tilde{k}}\}$ is at least $1/2$, i.e.,
\begin{align}
\label{eq:k_tilde}
v^V_{e_{l_{\tilde{k}},\tilde{k}}} &\geq 1/2, &
v^V_{e_{l_{\tilde{k}}',\tilde{k}}} &\geq 1/2, & &\text{and} & v^V_{e_{l_{\tilde{k}}'',\tilde{k}}} &\geq 1/2.
\end{align}
For a contradiction, let us assume that for each clause~$k = l_k \vee l_k' \vee k_k''$ there is a literal $l^*_k \in
\{l_k,l_k',l_k''\}$ such that $\smash{v^V_{e_{l^*_k,k}} < 1/2}$. As each variable $x_i \in V(\phi)$ splits its unit demand
between the path consisting of the positive literal edges $\{e_{x_i,k} : k \in K(\phi)\}$ and the path consisting of the
negative literal edges $\{e_{\bar{x}_i,k} : k \in K(\phi)\}$, at most one of these two paths is used
with a flow strictly smaller than $1/2$. Thus, the assignment vector $\vec y$ defined as 
\begin{align*}
y_{x_i} =
\begin{cases}
\true,  &\text{ if } v^V_e < 1/2 \text{ for all } e \in \{e_{x_i,k} : k \in K(\phi)\},\\
\false, &\text{ if } v^V_e < 1/2 \text{ for all } e \in \{e_{\bar{x}_i,k} : k \in K(\phi)\},\\
\true,  &\text{ otherwise,}
\end{cases}
\end{align*}
is well-defined. By construction, $\vec y$ satisfies all clauses, which is a contradiction to the assumption that no such
assignment exists. We conclude that there is a clause $\tilde{k}$ such that \eqref{eq:k_tilde} holds.

We proceed to bound the total cost of a solution. As $\vec v$ is a Wardrop equilibrium in which the clause commodity
$j_{\tilde{k}}$ uses at least partially the clause edge $e_{\tilde{k}}$, we further derive that $\sum_{e \in E_{\tilde{k}}}
v_{e}/z_{e} \geq v_{e_{\tilde{k}}}/z_{e_{\tilde{k}}} > 4$. We bound the total cost of the solution $(\vec v, \vec z)$ by
observing
\begin{align*}
C(\vec v, \vec z)
&= \sum_{e \in E_L} \bigl( v_e^2/z_e + z_e \bigr) + \sum_{e \in E_K} \bigl( (4+v_e/z_e)v_e + (\epsilon/2)^2 z_e\bigr)\\
&\geq \sum_{e \in E_L}\overline{\min}_{z_e \geq0} \bigl(v_e^2/z_e + z_e \bigr) + \sum_{e \in E_K}
\overline{\min}_{z_e\geq 0}\bigl( (4+v_e/z_e)v_e + (\epsilon/2)^2 z_e\bigr),
\end{align*}
where we slightly abuse notation by writing $\overline{\min}_{z_e \geq 0}$ shorthand for $\min_{z_e \geq 0 : \vec v \in
\mathcal{W}(z)}$. We obtain an upper bound by relaxing $\overline{\min}_{z_e \geq 0}$ to $\min_{z_e \geq 0}$ for the edges
in $E_L \setminus E_{\tilde{k}}$ and $E_K$. Hence,  
\begin{align*}
C(\vec v, \vec z) \geq \sum_{e \in E_L \setminus E_{\tilde{k}}} \min_{z_e \geq0} \bigl(v_e^2/z_e + z_e \bigr)
    + \sum_{e \in E_{\tilde{k}}} \overline{\min_{z_e \geq 0}} \bigl( v_e^2/z_e + z_e \bigr) +
\sum_{e \in E_K} \min_{z_e\geq
0}\bigl( (4+v_e/z_e)v_e + (\epsilon/2)^2 z_e\bigr).
\end{align*}
Calculating the respective minima, we obtain
\begin{align}
\label{eq:expression0}
C(\vec v, \vec z) &=    \sum_{e \in E_L \setminus E_{\tilde{k}}} 2v_e + \sum_{e \in E_{\tilde{k}}}
\underbrace{\overline{\min_{z_e \geq0}} \bigl( v_e^2/z_e + z_e \bigr)}_{\geq 2v_e} + \!\!\sum_{e \in E_K} (4+\epsilon) v_e.
\end{align}
Each clause commodity $j_k$ can route its demand either over the clause edge $e_k$ or over the three literal edges in
$E_k$. Every fraction of the demand routed over the clause edge contributes $4+\epsilon$ to the expression on the right
hand side of \eqref{eq:expression0} while it contributes at least $6$ when routed over the literal edges. Thus, the right
hand side of \eqref{eq:expression0} is minimized when the clause commodities do not use the
literal edges at all. We then obtain

\begin{align*}
C(\vec v, \vec z)
&\geq \sum_{e \in E_L \setminus E_{\tilde{k}}} 2v_e^V + \sum_{e \in E_{\tilde{k}}} \overline{\min_{z_e \geq0}} \bigl(
(v_e^V)^2/z_e + z_e
\bigr) + (4+\epsilon)|E_K|\\
&= 2\Bigl(\kappa\nu-\sum_{e \in E_{\tilde{k}}} v_e^V\Bigr) + (4+\epsilon)\kappa+\sum_{e \in
E_k} \overline{\min_{z_e \geq0}} \bigl( (v_e^V)^2/z_e + z_e
\bigr),\\
&= 2\kappa\nu + (4+\epsilon)\kappa+\sum_{e \in
E_{\tilde{k}}} \overline{\min_{z_e \geq0}} \bigl( (v_e^V)^2/z_e + z_e  - 2v_e^V\bigr),\\
&> 2\kappa\nu + (4+\epsilon)\kappa+ Q,
\end{align*}
where $Q$ is the solution to the constrained minimization problem
\begin{align}
Q &= \min_{\substack{v^V_e,z_e> 0\\e \in E_{\tilde{k}}}} \sum_{e \in E_{\tilde{k}}} \bigl((v_e^V)^2/z_e +
z_e - 2v_e^V \bigr)\notag\\
&\text{s.t.: } \sum_{e \in E_{\tilde{k}}} v_{e}^V/z_e \geq 4 \label{eq:side_constraint1}\\
&\hphantom{\text{s.t.: } \sum_{e \in E_{\tilde{k}}}/z_e} v_{e}^V \geq 1/2 \text{ for all } e \in
E_{\tilde{k}}.\label{eq:side_constraint2}
\end{align}
Side constraint \eqref{eq:side_constraint1} is a relaxation of the requirement that $\vec v$ is a Wardrop
equilibrium as the latency of the literal edges is strictly larger than $4$. Side constraint \eqref{eq:side_constraint2} is
due to the fact that for clause $\tilde{k}$ the three corresponding literal edges $\smash{e_{l_{\tilde{k}},\tilde{k}}}$,
$\smash{e_{l_{\tilde{k}}',{\tilde{k}}}}$, and $\smash{e_{l_{\tilde{k}}'',\tilde{k}}}$ are used with a flow of at least $1/2$
by the variable commodities. The optimal solution to the constraint
optimization problem $Q$ is equal to $Q = 1/8$ and is attained for $v_e^V = 1/2$ and
$z_e = 3/8$ for all $e \in E_{\tilde{k}}$. This implies that the total cost of a solution is not smaller than
$(4+\epsilon)\kappa +
2\kappa\nu + 1/8$, which finishes the first case of this proof.

\emph{Second case:} There is a clause commodity $j_{\tilde{k}}$ that does not use its clause edge $e_{\tilde{k}}$,
i.e., $v_{e_{\tilde{k}}} = 0$. As for first case, we observe
\begin{align*}
C(\vec v, \vec z) &= \sum_{e \in E_L} \bigl( v_e^2/z_e + z_e \bigr) + \sum_{e \in E_K} (4v_e+v_e^2/z_e +
(\epsilon/2)^2z_e) \geq \sum_{e \in E_L} 2v_e + \sum_{e \in E_K} (4+\epsilon)v_e.\\
\intertext{Using that $j_{\tilde{k}}$ does not use its clause edge, we derive that the flow on the
literal edges amounts to $\nu\kappa + 3$ and we obtain}
C(\vec v, \vec z) &\geq 2(\kappa\nu + 3) + (4+\epsilon)(\kappa-1) = 2\kappa\nu+
(4+\epsilon)\kappa+2,
\end{align*}
which concludes the proof.
\end{proof}

In the following remark we discuss that although the hardness proof of Theorem~\ref{thm:hardness_mc} used edges with
zero latency, the hardness result continues to hold even if edges with zero latency are not allowed.

\begin{remark}
\label{rem:zero_latency}
The continuous network design problem is $\classNP$-hard in the strong sense, even if no edges with zero latency are
allowed.
\end{remark}

\begin{proof}[Sketch of proof.]
Let $M$ be an upper bound on the total cost of an optimal solution to a continuous network design problem constructed in the
proof
of Theorem~\ref{thm:hardness_mc} and let $E_0$ be the set of edges with zero latency. We replace each edge $e \in E_0$, $ e
= (s,t)$, $s,t \in V$ by an edge $e' = (s,t)$ with latency function $S_{e'}(v_{e'}/z_{e'}) = v_{e'}/z_{e'}$ and construction
cost $\ell_{e'} = (\frac{\epsilon}{2M})^2$. For each new edge $e'$, we introduce an additional commodity $i_{e'}$ with
source $s$, sink $t$ and demand $M$. To route the flow of commodity $i_{e'}$, each solution has to
buy a sufficient capacity for the edge $e'$. For $z_{e'} = 4M^3/\epsilon$ the additional total cost on edge $e'$ are
$\epsilon$. Thus, the routing cost and the total cost on the new edges can be made arbitrarily small. In conclusion, we can
approximate the behavior of edges with zero latency within arbitrary precision by edges with unbounded latency functions. 
\end{proof}

\subsection*{Proof of Theorem~\ref{thm:hardness_apx}}

\begin{proof}
We reduce from a symmetric variant of \textsc{4-OCC-MAX-3-SAT} which is $\classNP$-hard to approximate, see Berman et
al.~\cite{Berman03}. An
instance of \textsc{4-OCC-MAX-3-SAT}, is given by a Boolean formula~$\phi$ in conjunctive normal form with the property that
each
clause contains exactly three literals and each variable occurs exactly four times. The problem to determine the
maximal number of clauses that can be satisfied simultaneously is known to be $\classNP$-hard to approximate within a factor
of $1016/1015 - \delta \approx 1.00099-\delta$ for any $\delta>0$, even for the special case that each variable occurs
exactly twice as a positive literal and exactly twice as a negative literal, see a follow-up paper by the same
authors~\cite{Berman03symmetric}. 

Let us again denote by $V(\phi)$, $K(\phi)$, and $L(\phi)$ the set of variables, clauses and literals of $\phi$ and let
$\nu = |V(\phi)|$ and $\kappa = |K(\phi)|$. It is convenient to assume that $K(\phi) = \{1,\dots,\kappa\}$ and $V(\phi) = \{x_1,\dots,x_{\nu}\}$. As every variable occurs exactly four times and every clause contains exactly
three literals, we have $4\nu = 3\kappa$. 
We slightly adjust the construction in the proof of Theorem~\ref{thm:hardness_mc} to make use of the information that each
literal occurs in exactly two clauses. We proceed to explain the construction of an instance of CNDP relative to a fixed
parameter $\epsilon \in (0,1/8)$. For a literal $l \in L(\phi)$, let $k_l,k_l' \in K(\phi)$ be the 
clauses that contain the literal $l$. We introduce two \emph{literal edges} $e_{l,k_l}$ and $e_{l,k_l'}$ with latency
function $S_e(v_e/z_e) = v_e/z_e$ and construction cost $\ell_e = 1$. For each variable $x_i \in V(\phi)$, we introduce a
corresponding \emph{variable commodity} $j_{x_i}$ with source $s_{j_{x_i}}$, sink $t_{j_{x_i}}$ and demands $d_{x_i} = 1$
that may then either choose the path consisting of the edges $e_{x_i,k_{x_i}}$ and
$e_{x_i,k_{x_i}'}$ that correspond to the positive literal $x_i$ or the edges $e_{\bar{x}_i,k_{\bar{x}_i}}$ and
$e_{\bar{x}_i,k_{\bar{x}_i}'}$ that correspond to the negative literal $\bar{x}_i$. We construct the network such that in the directed path containing the edges $e_{x_i,k}$ and $e_{x_i,k'}$ the edge $e_{x_i,k}$ appears before the edge $e_{x_i,k'}$ if and only if $k < k'$, i.e., the corresponding clause $k$ has a smaller index than the respective clause $k'$. As in the proof
of Theorem~\ref{thm:hardness_mc}, for each clause $k \in K(\phi)$, $k = l_k \vee l_k' \vee l_k''$, we introduce a
\emph{clause edge} $e_k$ with latency $S_e(v_e/z_e) = 4 + v_e/z_e$ and construction cost $\ell_e =
(\epsilon/2)^2$. For each clause $k \in K(\phi)$, there is a \emph{clause commodity} $j_k$ with source $s_{j_k}$, sink
$t_{j_k}$ and demand $d_{j_k} = 1$. The clause commodity $j_k$ may choose either the clause edge $e_k$ or a path that
contains all the corresponding literal edges
$e_{l_k,k}, e_{l_k',k}, e_{l_k'',k}$. The set of literal edges and clause edges is denoted by $E_L$ and $E_K$, respectively.
For notational convenience, for a clause $k \in K(\phi), k = l_k \vee l_k' \vee l_k''$, we
set $E_k = \{e_{l_k,k}, e_{l_k',k}, e_{l_k'',k}\}$. We add some additional edges with zero latency to obtain a network, see
Figure~\ref{fig:hardness_apx}. Because in each path for a variable commodity the clauses appear in increasing order of their index, adding these additional edges with zero latency does not add any further paths to the literal or variable commodities.

The hardness result continues to hold, even if edges with zero latency are not allowed, see
Remark~\ref{rem:zero_latency} after the proof of Theorem~\ref{thm:hardness_mc}.

\begin{figure}[bt]
\begin{center}
\begin{tikzpicture}[x=0.7cm,y=0.7cm]
\path[use as bounding box] (-1,-1.5) rectangle (21.5, 3.5);
\newcommand{\vargadgetapx}[4]{
\footnotesize
\path (#2,#3) node (v#10) [terminal,label=left:$s_{j_{x_{#4}}}$] {}
to  ++(60:1.5) node (v#1r1) {}
--  ++(0:1.5) node (v#1r2) {}
--  ++(0:1.5) node (v#1r3) {}
--  ++(-60:1.5) node (v#14) [terminal,label=right:$t_{j_{x_{#4}}}$] {}
    ++(-120:1.5) node (v#1l3) {}
--  ++(180:1.5) node (v#1l2) {}
--  ++(180:1.5) node (v#1l1) {}
--  ++(120:1.5) node (v#10) [terminal] {};
\draw[dashed,->] (v#10) -- (v#1r1);
\draw[->] (v#1r1) -- (v#1r2);
\draw[->] (v#1r2) -- (v#1r3);
\draw[dashed,->] (v#1r3) -- (v#14);
\draw[dashed,->] (v#10) -- (v#1l1);
\draw[->] (v#1l1) -- (v#1l2);
\draw[->] (v#1l2) -- (v#1l3);
\draw[dashed,->] (v#1l3) -- (v#14);
}

\vargadgetapx{1}{0}{0}{1}

\vargadgetapx{2}{7.5}{0}{2}
\node at (14,0) [draw=none,fill=none] {$\dots$};
\vargadgetapx{n}{15.5}{0}{\nu}
\clausegadget{1}{2}{3}{1}
\clausegadget{2}{8}{3}{2}
\node at (13.5,3) [draw=none,fill=none] {$\dots$};
\clausegadget{n}{16}{3}{\kappa}
\draw[->,aux,dashed] (sk1)  to[bend right=10] (v1r1);
\draw[->,aux,dashed] (v1r2) .. controls (5,3) and (7,-2) .. (v2l1);
\draw[->,aux,dashed] (v2l2) .. controls (13,-4) and (15,2.7) .. (vnr1);
\draw[->,aux,dashed] (vnr2) .. controls (15,3) and (6,1) .. (tk1);
\draw[->,dot,aux] (sk2) to[bend right=10] ++(-130:1) node
(dummy1)[draw=none,fill=none] {};
\draw[<-,dot,aux] (tk2) to[bend left=10] ++(-50:1) node
(dummy2)[draw=none,fill=none] {};
\draw[->,dot,aux] (skn) to[bend right=10] ++(-130:1) node
(dummy1)[draw=none,fill=none] {};
\draw[<-,dot,aux] (tkn) to[bend left=10] ++(-50:1) node
(dummy2)[draw=none,fill=none] {};
\end{tikzpicture}
\end{center}
\caption{
\label{fig:hardness_apx}
Network used to show the $\classAPX$-hardness of the continuous network design
problem. Clause~$1$ is equal to $x_1 \vee \bar{x}_2 \vee x_{\nu}$. Dashed edges have zero latency.}
\end{figure}

We claim that the so-defined instance of CNDP has a solution with total cost in the interval
\begin{align*}
\Bigl[10\kappa + |\tilde{K}|/4, (10+\epsilon)\kappa + (1/4+\epsilon/2)|\tilde{K}|\Bigr]
\end{align*}
if and only if the minimum number of unsatisfied clauses is $|\tilde{K}|$.

First, we show that an optimal solution has total cost not larger than $(10+\epsilon)\kappa + |\tilde{K}|/4$ if $\phi$ has a
solution $\vec y$ that violates $|\tilde{K}|$ clauses only. To this end, let $\vec y = (y_{x_i})_{x_i \in V(\phi)}$ be such
a
solution and let $\tilde{K}$ be the set of clauses that is not satisfied by $\vec y$. Consider the tuple $(\vec v, \vec
z)$ defined as
\begin{align*}
z_{e_{l,k}} &=
\begin{cases}
1,   &\text{if } k \notin \tilde{K} \text{ and } l = \neg y_i \text{ for some } i \in V(\phi),\\
\frac{1}{4/3 + \epsilon/6}, &\text{if } k \in \tilde{K} \text{ and } l = \neg y_i \text{ for some } i \in V(\phi),\\
0,   &\text{otherwise,}
\end{cases}
& &\text{for all } l \in L(\phi), k \in \{k_l,k_l'\},\\
v_{e_{l,k}} &=
\begin{cases}
1,   &\text{\phantom{/4}if } l = \neg y_i \text{ for some } i \in V(\phi),\\
0,   &\text{\phantom{/4}otherwise,}
\end{cases}
& &\text{for all } l \in L(\phi), k \in \{k_l,k_l'\},\\
z_{e_k} &= 2/\epsilon,
& &\text{for all } k \in K(\phi).\\
v_{e_k} &= 1,
& &\text{for all } k \in K(\phi).
\end{align*}

First, we show that the tuple $(\vec v, \vec z)$ is a solution to CNDP. To this end, it suffices to prove that $\vec v$ is a
Wardrop equilibrium for the latency functions defined by $\vec z$. We will argue for each commodity separately that it only uses shortest paths, starting with an arbitrary clause commodity $j_k$ that
corresponds to a non-satisfied clause $k \in \tilde{K}$, $k = l_k \vee l_k' \vee l_k''$. Such a clause uses the clause edge
$e_k$ with latency $4+\epsilon/2$. On the other hand, the corresponding literal edges $e_{l_k,k}$, $e_{l_k',k}$, and $e_{l_k'',k}$ have
capacity $\frac{1}{4/3+\epsilon/6}$ and carry one unit of flow of the variable commodities. Thus, their latencies sum up to $4+\epsilon/2$, implying that clause commodity $j_k$ is in equilibrium.
Next, consider a clause commodity $j_k$ that corresponds to a satisfied
clause $k \in K(\phi) \setminus \tilde{K}$, $k = l_k \vee l_k' \vee l_k''$. As $k$ is satisfied by $\vec y$, there is a
literal $l_k^* \in \{l_k, l_k', l_k''\}$ such that $l_k^* = y_{x_i}$ for some $x_i \in V(\phi)$. This implies that
$z_{e_{l_k^*,k}} = 0$ and, thus, edge $e_{l_k^*,k}$ has infinite latency. We derive that clause commodity $j_k$ has a unique path of
finite latency and this path is used in $\vec v$. Finally, consider a variable commodity $j_{x_i}$, $x_i \in V(x_i)$. As we
buy either the capacity for the edges $\smash{\{e_{x_i,k_{x_i}}, e_{x_i,k_{x_i}'}\}}$ corresponding to the
positive literal or the edges $\smash{\{e_{\bar{x}_i,k_{\bar{x}_i}}, e_{\bar{x}_i,k_{\bar{x}_i}'}\}}$ that
correspond to the negative literal, but not both, commodity $j_{x_i}$ has only one path with finite latency, and it uses that path in $\vec v$.

We proceed to calculate the total cost of the solution $(\vec v, \vec z)$. Every literal edge that corresponds to a
satisfied clause
and the negation of a literal in $y_i$ has capacity $1$ and flow $1$ and thus causes a total cost of $2$. In contrast to
this, each literal edge that corresponds to a violated clause and the negation of a literal in $y_i$ has capacity $\frac{1}{4/3 + \epsilon/6}$ and flow $1$ and, thus, causes a total cost of
\begin{align*}
\frac{4}{3} + \frac{\epsilon}{6} + \frac{1}{4/3 + \epsilon/6} \leq \frac{25}{12} + \frac{\epsilon}{6}.
\end{align*}
Further, each clause edge has capacity $2/\epsilon$ and is used by
$1$ unit of flow and, thus, contributes $4+\epsilon$ to the total cost. We calculate
\begin{align*}
C(\vec v, \vec z) &\leq 3 \bigl( 2(\kappa-|\tilde{K}|) + (25/12 + \epsilon/6)|\tilde{K}|\bigr) + (4+\epsilon)\kappa\\
&= (10+\epsilon)\kappa + (1/4 + \epsilon/2)|\tilde{K}|.
\end{align*}

We proceed to prove that an optimal solution of CNDP has total cost not smaller than $10\kappa +
\tilde{K}/4$ if each solution $\vec y$ of $\phi$ violates at least $\tilde{K}$ clauses. To this end, we need some
additional notation. For an edge flow $\vec v$, let $\vec v^V$ denote the edge flow that is due to the variable commodities
and $\vec v^K$ denote the edge flow that is due to the clause commodities. For a clause $k \in K(\phi)$, let $m_k(\vec v^V)
= \min_{e \in E_k} v^V_e$. In addition, we
set $\bar{K}(\vec v^K) = \{k \in K : \vec v^K_{e_k} > 0\}$, i.e., $\bar{K}(\vec v^K)$ is the set of clauses $k$ that uses
(at least partially) its clause edge $e_k$.

We bound $C(\vec v, \vec z)$ by observing
\begin{align*}
C(\vec v, \vec z)
&= \sum_{e \in E_L} \bigl( v_e^2/z_e + z_e \bigr) + \sum_{e \in E_K} \bigl( (4+v_e/z_e)v_e + (\epsilon/2)^2 z_e\bigr)\\
&\geq \sum_{e \in E_L} \overline{\min_{z_e \geq0}} \bigl(v_e^2/z_e + z_e \bigr)
     + \sum_{e \in E_K} \overline{\min_{z_e\geq 0}} \bigl( (4+v_e/z_e)v_e + (\epsilon/2)^2 z_e\bigr),
\end{align*}
where we again slightly abused notation writing $\overline{\min}_{z_e \geq 0}$ shorthand for $\min_{z_e \geq 0 : \vec v \in \W(z)}$.
We obtain a lower bound on the total cost observing that the latency of the clause edges is at least $4$. Thus,
\begin{align}
\label{eq:expression}
C(\vec v, \vec z) \geq \sum_{e \in E_L} \underbrace{\overline{\min_{z_e \geq0}} \bigl(v_e^2/z_e + z_e \bigr)}_{\geq 2} +
4\sum_{e \in E_K} v_e
\end{align}
Every unit of flow of a clause commodity $j_k$ with $k \in \bar{K}(\vec v^K)$, $k=l_k \vee l_k' \vee l_k''$ contributes at least $6$
to the right hand side of \eqref{eq:expression} when routed over the corresponding literal edges $e_{l_k,k}$, $e_{l_k',k}$,
 and
$e_{l_k'',k}$, but contributes only $4$ when routed over the corresponding clause edge $e_k$. Thus, we
obtain a lower bound assuming that each clause commodities $j_k$, $\smash{k \in \bar{K}(\vec v^K)}$ exclusively uses its clause edge, i.e.,
\begin{align}
\label{eq:expression2}
C(\vec v, \vec z)
&\geq 4|\bar{K}(\vec v^V)| + \sum_{k \in K(\phi)} \sum_{e \in E_k} \overline{\min_{z_e \geq0}} \bigl( (v_e^V +
v_e^K)^2/z_e + z_e \bigr).
\end{align}
With the same arguments, we observe that every clause commodity $j_k$ with $k \in K(\phi) \setminus \bar{K}(\vec v^K)$ contributes at least
$6$ to the right hand side of \eqref{eq:expression2} when routed over the literal edges, but contributes only $4$
when routed over
the clause edge. Thus, we obtain a lower bound assuming that $K(\phi) = \bar{K}(\vec v^K)$, i.e., every clause commodity
$j_k$ routes its demand exclusively over the corresponding clause edge $e_k$. Then,
\begin{align*}
C(\vec v, \vec z)
&\geq 4\kappa + \sum_{k \in K(\phi)} \sum_{e \in E_k} \overline{\min_{z_e \geq0}} \bigl(v_e^V)^2/z_e
+ z_e \bigr)\\
&= 4\kappa +  \!\!\!\sum_{\substack{k \in K(\phi)\\m_k(\vec v^V) = 0}} \sum_{e \in E_k} \overline{\min_{z_e \geq0}}
\bigl(v_e^V)^2/z_e + z_e \bigr) + \!\!\!\sum_{\substack{k \in K(\phi)\\m_k(\vec v^V) > 0}} \sum_{e \in E_k}
\overline{\min_{z_e \geq0}} \bigl(v_e^V)^2/z_e + z_e \bigr)
\end{align*}
Note that for all clauses $k \in K(\phi)$ with $m_k(\vec v^V) = 0$ at least one of the corresponding clause edges is not
used by the variable commodities and, thus, we can set the capacity of this edge to $0$. This implies that the
corresponding clause commodity $k$ stays at its clause edge and we can optimize the capacity of the remaining
edges in $E_k$ irrespective of the equilibrium constraints. We obtain
\begin{align*}
C(\vec v, \vec z)
&\geq 4\kappa + 6|k \in K(\phi) : m_k(\vec v^V) = 0| + \!\!\sum_{\substack{k \in K(\phi)\\m_k(\vec v^V) > 0}}
\sum_{e \in E_k} \overline{\min_{z_e \geq0}} \bigl( (v_e^V)^2/z_e + z_e\bigr)\\
&\geq 10\kappa + \!\!\!\sum_{\substack{k \in K(\phi)\\m_k(\vec v^V) > 0}} \sum_{e \in E_k} \overline{\min_{z_e \geq0}}
\bigl( (v_e^V)^2/z_e + z_e - 2v_e^V\bigr)\\
&\geq 10\kappa + \sum_{k \in K(\phi)} Q_k,
\intertext{where each $Q_k$ is the solution to the constrained minimization problem}
Q_k &= \min_{\substack{v^V_e,z_e> 0\\e \in E_k}} \sum_{e \in E_k} \bigl((v_e^V)^2/z_e +
z_e - 2v_e^V \bigr)\\
&\text{s.t.: } \sum_{e \in E_k} v_{e}^V/z_e \geq 4\\
&\hphantom{\text{s.t.: } \sum_{e \in E_k}/z_e} v_{e}^V \geq m_k \text{ for all } e \in E_k.
 v\end{align*}
The optimal solution to this problem is equal to $Q_k = m_k(\vec v^V)/4$ and is attained for $v_e^V = m_k(\vec v^V)$ and
$z_e = 3m_k(\vec v^V)/4$ for all $e \in E_k$, $k \in K(\phi)$. We obtain
\begin{align*}
C(\vec v, \vec z) \geq 10\kappa+ \sum_{k \in K(\phi)} m_k(\vec v^V)/4.
\end{align*}

To finish the proof it suffices to show that $\sum_{k \in K(\phi)} m_k(\vec v^V) \geq \tilde{K}$ for each flow of the variable
commodities $\vec v^V$. To this end, let $\vec v^V$ be a flow that minimizes $\sum_{k \in K(\phi)} m_k(v^V)$. We claim that
it is without loss of generality to assume that $\vec v^V$ is integral. To see this claim, suppose that the flow for all variable commodities except $j_{x_i}$ is fixed and consider the variable commodity $j_{x_i}$. Let $p$ denote the portion of the flow sent over the path
consisting of the positive literal edges $e_{x_i,k_{x_i}}$ and $e_{x_i,k_{x_i}'}$. By definition, only the clauses $k_{x_i}$ and $k_{x_i'}$ contain the literal $x_i$ and only the clauses $k_{\bar{x}_i}$ and $k_{\bar{x}_i}'$ contain the literal $\bar{x}_i$. Then, we can calculate the contribution of these four clauses to $\sum_{k \in K(\phi)} m_k(\vec v^v)$ as follows:
\begin{align*}
\sum_{k \in \{k_{x_i},k_{x_i}',k_{\bar{x}_i},k_{\bar{x}_i}'\}} m_k(\vec v^V) &= \sum_{k \in \{k_{x_i},k_{x_i}',k_{\bar{x}_i},k_{\bar{x}_i}'\}} \min_{e \in E_k} v^V_e \\
&= \sum_{k \in \{k_{x_i},k_{x_i}'\}} \min\Bigl(p,
\min_{e \in E_{k} \setminus e_{x_i,k}} \vec v^V_e\Bigr) + \sum_{k \in \{k_{\bar{x}_i},k_{\bar{x}_i}'\}} \min\Bigl(1-p,
\min_{e \in E_{k} \setminus e_{\bar{x}_i,k}} \vec v^V_e\Bigr)
\end{align*}

For a fixed flow $\vec v^V$ on the literal edges not involving $x_i$, this expression is concave in $p$. Hence, the minimum is
attained for either $p=0$ or $p=1$. Put differently, for any flow of the other variable commodities, the expression $\sum_{k \in K(\phi)} m_k(\vec v^V)$ is minimized when variable commodity $j_{x_i}$ routes all of its demand on one path. Iterating this argument for all variable commodities, we conclude that is without loss of generality to assume that $\vec v^V$ is integral.

For an integral flow $\vec v^V$ of the variable commodities, consider the $\true/\false$ assignment $\vec y = (y_{x_i})_{x_i \in
V(\phi)}$ defied as $y_i = \true$ if and only if $v_{y_i,k_{y_i}} = 0$. As this assignment satisfies at most $K^*$ clauses,
we have that $\sum_{k \in K(\phi)} m_k(v^V) \geq \tilde{K}$.

Plugging everything together, we obtain that the total cost of an optimal solution to CNDP lies in the range
\begin{align}
\label{eq:interval}
\Bigl[10\kappa + |\tilde{K}|/4,\,(10+\epsilon)\kappa + (1/4 + \epsilon/2)|\tilde{K}|\Bigr]
\end{align}
if $|\tilde{K}|$ clauses cannot be satisfied. 

Berman et al.~\cite{Berman03,Berman03symmetric} construct a family of symmetric instances of \textsc{4-OCC-MAX-3-SAT} with
$\kappa = 1016n$, $n \in \N$ that has the property that for any $\delta \in (0,1/2)$ it is $\classNP$-hard to distinguish
between the systems where $(1016-\delta)n$ clauses can be satisfied and systems where at most $(1015+\epsilon)n$ clauses can
be satisfied. Using \eqref{eq:interval}, the corresponding instances of CNDP have the property that they have total cost at most
$(10+\epsilon) 1016n +\delta n(\frac{1}{4} + \frac{\epsilon}{2})$, if at least $(1016-\delta)n$ clauses can be satisfied,
and total cost at least $10 \cdot 1016n + 1/4 - \frac{\delta n}{4}$, if at most $(1015+\delta)n$ clauses can be satisfied.
As we let
$\epsilon$ and $\delta$ go to zero, we derive that it is $\classNP$-hard to approximate CNDP by any factor better than
$10160.25/10160 \approx 1.000024$. This proves the 
$\classAPX$-hardness of the problem.
\end{proof}

\subsection*{Hardness for undirected networks}

\begin{theorem}
\label{thm:hardness_apx_un}
The continuous network design problem on undirected networks is $\classAPX$-hard, even if all latency functions are affine.
\end{theorem}

\begin{figure}[bt]
\begin{center}
\begin{tikzpicture}[x=0.7cm,y=0.7cm]
\path[use as bounding box] (-1,-1.5) rectangle (21.5, 3.5);
\newcommand{\vargadgetapxun}[4]{
\footnotesize
\path (#2,#3) node (v#10) [terminal,label=left:$s_{j_{x_{#4}}}$] {}
to  ++(60:1.5) node (v#1r1) {}
--  ++(0:1) node (v#1r2) {}
--  ++(0:1) node (v#1r2b) {}
--  ++(0:1) node (v#1r3) {}
--  ++(-60:1.5) node (v#14) [terminal,label=right:$t_{j_{x_{#4}}}$] {}
    ++(-120:1.5) node (v#1l3) {}
--  ++(180:1) node (v#1l2b) {}
--  ++(180:1) node (v#1l2) {}
--  ++(180:1) node (v#1l1) {}
--  ++(120:1.5) node (v#10) [terminal] {};
\draw[dashed] (v#10) -- (v#1r1);
\draw (v#1r1) -- (v#1r2);
\draw[dashed] (v#1r2) -- (v#1r2b);
\draw (v#1r2b) -- (v#1r3);
\draw[dashed] (v#1r3) -- (v#14);
\draw[dashed] (v#10) -- (v#1l1);
\draw (v#1l1) -- (v#1l2);
\draw[dashed] (v#1l2) -- (v#1l2b);
\draw (v#1l2b) -- (v#1l3);
\draw[dashed] (v#1l3) -- (v#14);
}
\newcommand{\vargadgetapxunano}[4]{
\footnotesize
\path (#2,#3) node (v#10) [terminal,label=left:$s_{j_{x_{#4}}}$] {}
to  node [lbl,left] {type one~~~~~} ++(60:1.5) node (v#1r1) {}
--  ++(0:1) node (v#1r2) {}
--  ++(0:1) node (v#1r2b) {}
--  ++(0:1) node (v#1r3) {}
--  ++(-60:1.5) node (v#14) [terminal,label=right:$t_{j_{x_{#4}}}$] {}
    ++(-120:1.5) node (v#1l3) {}
--  ++(180:1) node (v#1l2b) {}
--  ++(180:1) node (v#1l2) {}
--  ++(180:1) node (v#1l1) {}
--  ++(120:1.5) node (v#10) [terminal] {};
\draw[dashed] (v#10) -- (v#1r1);
\draw (v#1r1) -- (v#1r2);
\draw[dashed] (v#1r2) to node [lbl,below] {type two} (v#1r2b);
\draw (v#1r2b) -- (v#1r3);
\draw[dashed] (v#1r3) -- (v#14);
\draw[dashed] (v#10) -- (v#1l1);
\draw (v#1l1) -- (v#1l2);
\draw[dashed] (v#1l2) -- (v#1l2b);
\draw (v#1l2b) -- (v#1l3);
\draw[dashed] (v#1l3) -- (v#14);
}

\vargadgetapxunano{1}{0}{0}{1}

\vargadgetapxun{2}{7.5}{0}{2}
\node at (14,0) [draw=none,fill=none] {$\dots$};
\vargadgetapxun{n}{15.5}{0}{\nu}
\clausegadgetun{1}{2}{3}{1}
\clausegadgetun{2}{8}{3}{2}
\node at (13.5,3) [draw=none,fill=none] {$\dots$};
\clausegadgetun{n}{16}{3}{\kappa}
\node at (2.8,2.7) [draw=none,fill=none] {type five};
\draw[aux,dashed] (sk1)  to[bend right=10] node [lbl,left] {type three~~~~~} (v1r1);
\draw[aux,dashed] (v1r2) .. controls (5,3) and (7,-2) .. (v2l1);
\node at (5,1.6) [draw=none,fill=none] {type four};
\draw[aux,dashed] (v2l2) .. controls (13,-4) and (15,2.7) .. (vnr1);
\draw[aux,dashed] (vnr2) .. controls (15,3) and (6,1) .. (tk1);
\draw[dot,aux] (sk2) to[bend right=10] ++(-130:1) node
(dummy1)[draw=none,fill=none] {};
\draw[dot,aux] (tk2) to[bend left=10] ++(-50:1) node
(dummy2)[draw=none,fill=none] {};
\draw[dot,aux] (skn) to[bend right=10] ++(-130:1) node
(dummy1)[draw=none,fill=none] {};
\draw[dot,aux] (tkn) to[bend left=10] ++(-50:1) node
(dummy2)[draw=none,fill=none] {};
\end{tikzpicture}
\end{center}
\caption{
\label{fig:hardness_apx_un}
Network used to show the $\classAPX$-hardness of the continuous network design
problem on undirected graphs. Clause~$1$ is equal to $x_1 \vee \bar{x}_2 \vee x_{\nu}$. The clause edges (straight lines in
the upper part of the graph) and the literal edges (straight edges in the lower part of the graph) are connected via
different auxiliary edges (dashed). The auxiliary edges have different constant latencies dependent on their type. Type one
edges are auxiliary edges adjacent to a source or a target of a variable commodity. Type two edges are auxiliary edges
connecting two literal edges that correspond to the same literal. Type three edges are auxiliary edges adjacent to the a
source or a target of a clause commodity. Type four edges are auxiliary edges connecting two literal edges that correspond
that correspond to different literals that appear together in a clause. Type five edges connect the source of a clause
commodity with the respective clause edge.}
\end{figure}
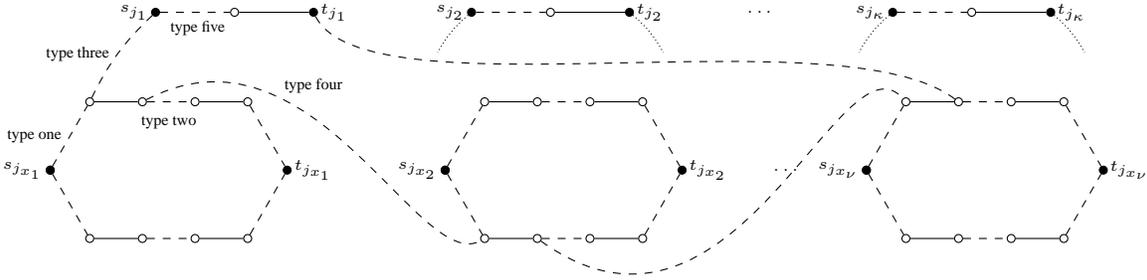

\begin{proof}[Sketch of proof]
As in the proof of Theorem~\ref{thm:hardness_apx}, we reduce from a symmetric variant of 
\textsc{4-OCC-MAX-3-SAT} where each variable occurs exactly twice negated and twice unnegated. 
We will closely mimic the  proof of Theorem~\ref{thm:hardness_apx} and only
sketch how to adjust it to the undirected case.

We use a construction similar to the directed case, see Figure~\ref{fig:hardness_apx_un}. We carefully choose the latency of
the auxiliary edges in order to prevent the commodities from taking undesired paths. For each variable commodity $j_{x_i}$,
let us call the two dashed edges adjacent to $s_{j_{x_i}}$ and the two edges adjacent to $t_{j_{x_i}}$ \emph{type one
edges}. Further, let us call the dashed edge between the edges $e_{x_i,k}$ and $e_{x_i,k'}$ and between $e_{\bar{x}_i,k}$
and $e_{\bar{x}_i,k'}$ \emph{type two edges}. For each clause commodity $j_i$, we call the dashed edge connecting $s_{j_i}$ the a variable gadget
and the dashed edge adjacent to $t_{j_i}$ \emph{type three edges}. We call the dashed edges connecting two literal edges
corresponding to different variables but the same clause \emph{type four edges}. Finally, we call the dashed edges that
connect the source node of a clause commodity with the respective clause edge \emph{type five edges}.

We set the latency of the type one edges to $50$, of the type two edges to $100$, of the type three edges to $0$, of the type four edges to $20$, and of the type five edges to $40$.

 We claim that the total cost of an optimal solution to CNDP lies in the range
\begin{align}
\label{eq:interval_apx}
\Bigl[200\kappa + |\tilde{K}|/4,\,(200+\epsilon)\kappa + (1/4 + \epsilon/2)|\tilde{K}|\Bigr]
\end{align}
if exactly $|\tilde{K}|$ clauses cannot be satisfied.

To see the upper bound in \eqref{eq:interval_apx}, fix an assignment of the variables that satisfies $\kappa-
|\tilde{K}|$ clauses and construct a solution to CNDP analogously to the proof of the directed case, i.e., route all
clause commodities along the clause edges, all variable edges along the negation of the assignment of the variable and
choose the installed capacities as in the proof of Theorem~\ref{thm:hardness_apx}. We will show that with these capacities 
the constructed flow is a Wardrop equilibrium. Since the auxiliary edges have
non-zero latency, compared to the solution in the directed case, the latency cost of each clause commodity increased by $40$
and the latency cost of each variable commodity increased by $200$. Thus, the total cost increased by $40\kappa +
\frac{3}{4} \cdot 200 \kappa = 190 \kappa$ giving a total cost of $(200+\epsilon)\kappa + (1/4 + \epsilon/2)|\tilde{K}|$. It is
left to argue that this solution still constitutes a Wardrop equilibrium although all edges can now be used in both
directions. To this end, note that each clause commodity uses its clause edge and experiences a 
total latency of $40 + 4 + \epsilon/2$ = $44 + \epsilon/2$. However, each other path available to a clause commodity uses
either a type two edge (with latency $100$), \emph{or} two type three edges, two type four edges (each with latency $20$), and the three corresponding literal edges (with latencies summing up to $4+\epsilon/2$, as before).
Thus, no clause commodity wants to deviate to another path and the constructed solution is a Wardrop equilibrium analogously to the directed
case.

For the lower bound, we argue as follows. If no variable commodity uses a type three edge or a type four edge, then each variable commodity has to split its flow between the path corresponding to the positive and the negative literal, respectively, and the lower bound can be proven analogously to the directed case.

So we are left with cases that a variable commodity uses a type three edge or a type four edge. Let us first assume that we
have an optimal solution, in which a variable commodity uses a type four edge. We may assume without loss of generality that every literal edge that carries flow has a latency of at most $5$, because we could decrease the total cost by increasing the capacity on these edges, otherwise. (However, we may not decrease the latency below $4 + \epsilon/2$ because this might give an incentive to the clause commodities to use these edges as well.) Every path available to a variable
commodity uses at least two type one edges as these edges are adjacent to the source and target of each variable commodity.
It is also not hard to see that every path available to a variable commodity has to use at least either two additional type
one edges or one type two edge. Using that the variable commodity also uses a type four edge, this implies that the latency
of the variable commodity is at least $200+20$. However, it would also be feasible to route that variable along the path
corresponding to the positive literal say while installing an additional capacity of $1/5$ on the two literal edges of the
positive literal resulting in a total cost of $200+10+2/5 < 220$. This low capacity would not prevent any of the clause
commodities from using their clause edge and has a lower total cost. Thus, we may conclude that no variable commodity uses a
type four edge. As any path of a variable commodity that uses a type three edge also uses a type five edge with latency
$40$, we may conclude that no variable commodity uses 
such an edge as well.
\end{proof}

\subsection*{Proof of Proposition~\ref{pro:single_sink}}

\begin{proof}
We solve the relaxed problem (CNDP'). As in the proof of Proposition~\ref{pro:marcotte}, for each edge $e
\in E$, we find a solution to the equation $x^2S_e'(x) = l_e$, which we denote by $u_e$. Then, we find an unsplittable flow
that minimizes
\begin{align}
\label{eq:marcotte_relax}
\min_{\vec v \in \flows} \sum_{e\in E} \bigl( S_e(u_e) + \ell_e/u_e\bigr) v_e,
\end{align}
Let $T$ be a shortest path tree routed in $t$ w.r.t.\ the edge weights $w_e = S_e(u_e) + \ell_e/u_e$. By
construction, each commodity $k$ has a unique path in $T$ that connects the source $s_k$ to the joint sink $t$. For each
$e
\in T$, let $d_e$ be the sum of the demands of the commodities that use edge $e$ in $T$ along its path. For each edge $e \in
T$ we buy capacity $z_e = d_e/u_e$ and route a flow of $d_e$. All other edges have zero capacity and, thus, infinite
latency. By construction, the total cost of this solution equals \eqref{eq:marcotte_relax}. Also, the resulting flow is a
Wardrop equilibrium as every commodity $k$ has a unique path from $s_k$ to $t$ that uses only edges with non-zero capacity.
\end{proof}

\subsection*{Proof of Lemma~\ref{lem:equivalent}}

\begin{proof}
The expression $\sup_{x \geq
0} \max_{\gamma \in [0,1]} \gamma\,\bigl(1-\frac{S(\gamma\,x)}{S(x)}\bigr)$ is non-negative and strictly positive for
$\gamma \in (0,1)$, thus, the inner maximum is attained for $\gamma \in (0,1)$. Hence, $\gamma$ satisfies the
first order optimality conditions
\begin{align*}
& & 0 &= \Bigl(1-\frac{S(\gamma\,x)}{S(x)}\Bigr) - \gamma\,x \cdot \frac{S'(\gamma\,x)}{S(x)}\\
&\Leftrightarrow & S(x) &= S(\gamma\,x) + \gamma\,x\,S'(\gamma\,x)
\end{align*}
By substituting $y = \gamma\,x$, we obtain
\begin{align*}
\sup_{x \geq 0} \max_{\gamma \in [0,1]} \gamma\,\Bigl(1-\frac{S(\gamma\,x)}{S(x)}\Bigr)
&= \sup_{y \geq 0} \Bigl\{\gamma
\,\Bigl(1 -\frac{S(y)}{S(y/\gamma)}\Bigr) : \gamma \in
[0,1] \text{ with } S(y/\gamma) = S(y) + S'(y)\,y\Bigr\}\\
&= \sup_{y \geq 0} \Bigl\{\gamma\cdot\frac{S'(y)\,y}{S(y) + S'(y)\,y} : \gamma \in
[0,1] \text{ with } S(y/\gamma) = S(y) + S'(y)\,y\Bigr\},
\end{align*}
which proves the lemma.
\end{proof}

\subsection*{Additional material for the proof of Theorem~\ref{thm:better}}
\begin{lemma}
\label{lem:maxof2}
For all $\gamma,\mu\in (0,1]$, we have
\begin{align}
\label{eq:maxof2}
\max_{p\in (0,1)}\min\left\{1 + \gamma(1-p), \Bigl(\sqrt{p} +
\sqrt{\mu(1-p)}\Bigr)^2\right\} 
&= 
\frac{(\gamma + \mu +1)^2}{(\gamma + \mu + 1)^2 - 4\mu\gamma} < 1+\mu.
\end{align}
\end{lemma}
\begin{proof}
%
%
Observe that $1 + \gamma(1-p)$ is decreasing in $p$. Elementary calculus shows that 
$\Bigl(\sqrt{p} + \sqrt{\mu(1-p)}\Bigr)^2$ attains its maximum at $p=\hat{p}:=\frac{1}{1+\mu}$, is increasing when $p< \hat{p}$ and decreasing afterwards.
Now, 
$\Bigl(\sqrt{\hat{p}} + \sqrt{\mu(1-\hat{p})}\Bigr)^2=1+\mu$ and $1+\gamma(1-\hat{p})= 1+\mu\frac{\gamma}{1+\mu}<1+\mu$, 
the inequality in \eqref{eq:maxof2} follows.

Moreover, it follows that the maximum on the left hand side of \eqref{eq:maxof2}  is attained for the unique $p^*
\in (0,\hat{p})$ such that $1 + \gamma(1-p^*)= \bigl(\sqrt{p^*} + \sqrt{\mu(1-p^*)}\bigr)^2$.
Thus, $p^*$ is a solution to the equation
\begin{align*}
0
&= -(1-p^*) - \gamma(1-p^*) + 2\sqrt{p^*(1-p^*)\mu} + \mu(1-p^*)\\
&= (1-p^*)\Bigl(2\sqrt{\mu \frac{p^*}{1-p^*}} + \mu - \gamma -1\Bigr)
\intertext{and since $p^* < 1$}
0 &= 2\sqrt{\mu \frac{p^*}{1-p^*}} + \mu - \gamma(\S) -1.
\end{align*}
The unique solution to this equation is
\begin{align*}
p^* = \frac{(\gamma-\mu+1)^2}{(\gamma-\mu+1)^2 + 4\mu}.
\end{align*}
Plugging this into the left hand side of \eqref{eq:maxof2} gives
\begin{align*}
\frac{(\gamma + \mu +1)^2}{(\gamma + \mu + 1)^2 - 4\mu\gamma},
\end{align*}
which proves the identity in \eqref{eq:maxof2}.
\end{proof}

\subsection*{Proof of Corollary~\ref{cor:simple} and Corollary~\ref{cor:better}}
\begin{proof}
For the proofs of Corollary~\ref{cor:simple} and Corollary~\ref{cor:better}, we give bounds on $\mu(\S)$ and
$\gamma(\S)$ for the respective sets $\S$ of allowable latency functions. Theorem~\ref{thm:bring-to-equilibrium},
Theorem~\ref{thm:scale-uniformly} and Theorem~\ref{thm:better} then give the respective approximation guarantees.

\paragraph{\it Arbitrary latency functions.}
First, we consider case {\it (a)} of both Corollaries, where $\S$ is a class of arbitrary non-negative and non-decreasing
latencies. We observe that
\begin{align*}
\mu(\S)=\sup_{S\in \S}\sup_{x \geq 0}
\max_{\gamma \in [0,1]} \gamma \Bigl(1- \frac{S(\gamma\,
x)}{S(x)}\Bigr) &\leq 1.
\end{align*}
By definition $\gamma(\S)\le 1$.
Now Corollary~\ref{cor:simple} {\it(a)} follows immediately and Corollary~\ref{cor:better} {\it(b)} follows from the fact
that 
\begin{align}
\label{eq:approximation_guarantee}
\frac{(\gamma(\S) + \mu(\S) +1)^2}{(\gamma(\S) + \mu(\S) + 1)^2 - 4\mu(\S)\gamma(\S)}
\end{align}
is strictly increasing in $\gamma(\S)$ and $\mu(\S)$. 

\paragraph{\it Concave latency functions.}
Next, consider case {\it(b)} of both Corollaries, where $\S$ contains concave latencies only.
Observe that
\begin{align*}
\mu(\S)=\sup_{S\in \S}\sup_{x \geq 0} \max_{\gamma \in [0,1]} \gamma \Bigl(1- \frac{S(\gamma\,x)}{S(x)}\Bigr) &\leq
\sup_{S\in \S}\sup_{x \geq 0}
\max_{\gamma \in [0,1]} \gamma \Bigl(1- \gamma - \frac{(1-\gamma)S(0)}{S(x)}\Bigr)\notag\\
&\leq \max_{\gamma \in [0,1]} \gamma (1- \gamma)\notag\\
&= 1/4,
\end{align*}
where the first inequality uses the concavity of all functions $S \in \S$. Further, as shown in Lemma~\ref{lem:equivalent},
the $\gamma$ for which the inner maximum is attained, satisfies the first order optimality conditions $S(x) = S(\gamma x) +
\gamma\,x\,S'(\gamma x)$. As $S$ is concave, we derive that $\gamma\,x\,S'(\gamma x) \leq S(\gamma x)$, which implies
\begin{align*}
S(x) \geq 2S(\gamma x) \geq 2(\gamma S(x) + (1-\gamma)S(0)) \geq 2\gamma S(x),
\end{align*}
and, thus, $\gamma(\S) \leq 1/2$. Again, Corollary~\ref{cor:simple} {\it(b)} follows immediately and
Corollary~\ref{cor:better} {\it(b)} follows from the fact that \eqref{eq:approximation_guarantee} is increasing in
$\gamma(\S)$ and $\mu(\S)$.

\paragraph{\it Polynomial latency functions.} 

Finally, consider case {\it (c)} of both Corollaries, where for some fixed
maximal degree $\Delta \geq 0$, the set $\S$ contains only polynomial latency functions of type $S(x) = \sum_{j=0}^{\Delta}a_j x^j$, with $a_j
\geq 0$ for all $j$.  Denote $\vec a=(a_j)_{j\in[0,\Delta]}$.
We calculate
\begin{align*}
\mu(\S)
&=\sup_{S\in \S}\sup_{x \geq 0} \max_{\gamma \in [0,1]} \gamma \Bigl(1- \frac{S(\gamma\,x)}{S(x)}\Bigr) \\
&=\sup_{\vec a\ge 0} \sup_{x \geq 0} \max_{\gamma \in [0,1]} \gamma \Bigl(1- \frac{\sum_{j=0}^{\Delta}a_j \gamma^j x^j}{\sum_{j=0}^{\Delta}a_j  x^j}\Bigr)\\
&=\sup_{\vec a\ge 0} \sup_{x \geq 0} \max_{\gamma \in [0,1]} \gamma \Bigl(\frac{\sum_{j=0}^{\Delta}a_j  x^j (1-\gamma^j)}{\sum_{j=0}^{\Delta}a_j  x^j}\Bigr)\\
\intertext{As $(1-\gamma^j)$ is increasing in $j$ for every $\gamma\in(0,1)$, it follows that the supremum 
over $\vec a\ge 0$ is attained if $a_\Delta>0$ and $a_j=0$ for all $j\in[0,\Delta-1]$. We get}
\mu(\S)
&= \max_{\gamma \in [0,1]} \gamma (1- \gamma^\Delta)\\
&= \Bigl(\frac{1}{\Delta+1}\Bigr)^{1/\Delta}\Bigl(1-\frac{1}{\Delta+1}\Bigr)\\
&=\Bigl(\frac{1}{\Delta+1}\Bigr)^{1/\Delta}\Bigl(\frac{\Delta}{\Delta+1}\Bigr),
\end{align*}
which directly implies the statement of Corollary~\ref{cor:simple} {\it (c)}. Further, this value is attained for $\gamma(\S) =
\bigl(\frac{1}{\Delta+1}\bigr)^{1/p}$. Plugging these values in \eqref{eq:approximation_guarantee} and rearranging terms,
we obtain the approximation guarantee claimed in Corollary~\ref{cor:better} {\it (c)}.
\end{proof}

\subsection*{Convex budget constraints}

\begin{theorem}\label{thm:budget}
 Let $\S$ be a class of latency functions.
\begin{enumerate}
\itemsep-0.2em
\item The following algorithm is a $\frac{1}{1-\mu(\S)}$-approximation for \eqref{budgets}\newline
(in particular a $4/3$-approximation for affine latencies):
\begin{enumerate}
\item Compute a solution $(\vec v^*, \vec z^*)$ to relaxation \eqref{budgets}.
\item Compute a Wardrop equilibrium $\vec w \in \W(\vec z^*)$.
\item Return $(\vec w,\vec z^*)$.
\end{enumerate}
\item\label{np-hard} For affine latencies,  there is no 
polynomial time approximation algorithm with a performance guarantee
better than $4/3-\epsilon$ for any $\epsilon>0$, unless $\classP= \classNP$.
\end{enumerate}
\end{theorem}

\begin{proof}
The upper bound in 1. is straight forward by using well known price of anarchy results
known in the literature, cf. Correa et al.~\cite{CorSS04}
and Roughgarden~\cite{Rough02} and Roughgarden and Tardos~\cite{Roughgarden02}.
For \ref{np-hard}., we mimic the construction put forward in Roughgarden~\cite{Roughgarden06}.

We reduce from the 2-Directed-Vertex-Disjoint-Paths (\DDP)
problem, which is strongly $\classNP$-complete. Given a directed graph $G = (V,E)$ and two node pairs $(s_1, t_1)$, $(s_2, t_2)$ the problem is to decide whether there exist a pair of \emph{vertex-disjoint} paths $P_1$ and $P_2$, where $P_1$ and $P_2$
are $(s_1,t_1)$ and $(s_2,t_2)$-paths, respectively.

We will show that a $(\frac{4}{3}-\epsilon)$-approximation algorithm
can be used to differentiate between ``Yes'' and ``No'' instances of \DDP\
in polynomial time. Given an instance $\mathcal{I}$ of \DDP we
construct a graph $G'$ by adding a super source $s$ and a super sink
$t$ to the network. We connect $s$ to $s_1$ and $s_2$ and $t_1$ and
$t_2$ to $t$, respectively. The latency functions of the added edges
are set to $S_e(v_e/z_e)=v_e/z_e$ for $e\in\{(s,s_1),(t_2,t)\}$
and  $S_e(v_e/z_e)=1+v_e/z_e$ for $e\in\{(s,s_2),(t_1,t)\}$.
The function $g(\vec z)$ assigns edge-specific budgets
according to  $B_{(s,s_1)}=1$ and $B_{(t_2,t)}=1$.
The per-unit cost of capacities are given by $\ell_e=1$ for
$e\in \{(s,s_1), (t_2,t)\}$ and $\ell_e=0$, otherwise.

We proceed to prove the following two
statements:
\begin{enumerate}
\item If $\mathcal{I}$ is a ``Yes'' instance of \DDP, then $G'$ admits
  a solution $(\vec v, \vec z)$ with $\vec v \in \W(\vec z)$ satisfying $C(\vec v, \vec z)\leq 3/2$.
\item If $\mathcal{I}$ is a ``no'' instance of \DDP, then
  $C(\vec v, \vec z)\geq 2$ for all $(\vec v, \vec z)$ with $\vec v \in \W(\vec z)$.
\end{enumerate}

To see the first statement, suppose $\mathcal{I}$ is a ``Yes'' instance and let $P_1$ and $P_2$ be the respective disjoint paths. For all edges contained in neither $P_1$ nor $P_2$, we install a capacity of $0$ leading to infinite latency of these edges.
For the edges in $P_1\cup P_2\cup\{(s,s_2), (t_1,t)\}$ we buy infinite capacity
resulting in $0$ latency on edges in $P_1\cup P_2$ and a latency of $1$
on $\{(s,s_2), (t_1,t)\}$. 
For the edges in $\{(s,s_1),(t_2,t)\}$
we spend the budgets of $1$ each.
Then, splitting the flow evenly along these paths
yields a Wardrop flow with routing cost $C(\vec z, \vec v)=2\cdot((1/2)^2+1/2\cdot 1)=3/2$.

To show the second statement, let $(\vec v, \vec z)$ be an optimal solution.
We may assume that there is an $(s,t)$ path.
  We consider the following cases.
\begin{enumerate}
\item For exactly one $i\in\{1,2\}$, all flow-carrying
paths contain the edges $(s,s_i)$ and $(t_i,t)$. For this case it is easy to see
that $C(\vec v, \vec z)\geq 2$ 
since all $4$ new edges have at least latency of $1$ if used with $1$ unit of flow.
\item There is a flow-carrying
path $P$ containing $(s,s_2)$ and $(t_1,t)$.
In this case, the latency along this
path is at least $2$, hence, since
every flow-carrying path has the same
latency, we obtain $C(\vec z, \vec v)\geq 2$.
\item There is a flow-carrying
path $P$ containing $(s,s_1)$ and $(t_2,t)$.
If all flow-carrying paths from
$s$ to $t$ contain $(s,s_1)$ and $(t_2,t)$,
we obtain $C(\vec z, \vec v)\geq 2$
using the budget constraints at $\{(s,s_1),(t_2,t)\}$. Suppose there
is another flow-carrying path $Q$ containing $(s,s_1)$
and  $(t_1,t)$. Then the latency on the subpath
$Q[s_1,t]$ must be at least $1$ and, by the Wardrop conditions,
the latency of $P[s_1,t]$ must be a least one.
If the entire demand uses edge $(s,s_1)$, the minimum
possible latency on this edge is $1$ and the
latency of $P$ (and also $Q$) must be at least two, thus, we obtain 
$C(\vec z, \vec v)\geq 2$.
Suppose, there is a flow-carrying path $R$ containing the edge $(s,s_2)$.
If $R$ contains edge $(t_1,t)$, we are in case 2.
Thus we may assume that $R$ contains edge $(t_2,t)$.
Since we are in a ``No'' instance of \DDP, the path $R$
must have one vertex with the path $Q$ in common
which implies that for $R[s_2,t]$ the latency is at least $1$ and, hence,
the latency of $R$ is at least $2$ giving
  $C(\vec v, \vec z)\geq 2$.
  \item The case that we have two flow-carrying $(s_1,t_1)$ and
  $(s_2,t_2)$ paths reduces to one of the cases 1., 2. or 3. since
  we are in a  ``No'' instance of \DDP.
\end{enumerate}

\end{proof}

\end{document}